\documentclass[aps,pra,twocolumn,superscriptaddress]{revtex4}

\usepackage{mathpazo}

\usepackage{amsmath}
\usepackage{amsfonts,amssymb,amsthm, bbm,braket}
\usepackage{graphicx}   
\usepackage{subfigure}  
\usepackage{amsbsy} 
\usepackage[bold]{hhtensor} 
\usepackage{natbib}
\usepackage{algpseudocode}

\usepackage{multirow}
\usepackage{tcolorbox}

\newtheorem{definition}{Definition}
\newtheorem{corollary}{Corollary}
\newtheorem{theorem}{Theorem}
\newtheorem{lemma}{Lemma}
\newtheorem{proposition}{Proposition}

\usepackage[breaklinks=true]{hyperref}

\hypersetup{
  colorlinks   = true, 
  urlcolor     = blue, 
  linkcolor    = blue, 
  citecolor   = red 
}

\begin{document}

\title{Fast state tomography with optimal error bounds}

\author{M.\ Guta}
\affiliation{School of Mathematical Sciences, University of Nottingham, United Kingdom}

\author{J.\ Kahn}
\affiliation{Institut de Math\'ematiques de Toulouse, Toulouse, France}

\author{R.\ Kueng}
\email[corresponding author: ]{rkueng@caltech.edu}
\affiliation{California Institute of Technology, Pasadena, United States}

\author{J.\ A.\ Tropp}
\affiliation{California Institute of Technology, Pasadena, United States}

\date{\today}


\begin{abstract}

Projected least squares (PLS) is an intuitive and numerically cheap technique for quantum state tomography.  The method first computes the least-squares estimator (or a linear inversion estimator) and then projects the initial estimate onto the space of states.  The main result of this paper equips this point estimator with a rigorous, non-asymptotic confidence region expressed in terms of the trace distance.  The analysis holds for a variety of measurements, including 2-designs and Pauli measurements.  The \emph{sample complexity} of the estimator is comparable to the strongest convergence guarantees available in the literature and---in the case of measuring the uniform POVM---saturates fundamental lower bounds.
The results are derived by reinterpreting the least-squares estimator as a sum of random matrices and applying a matrix-valued concentration inequality. The theory is supported by numerical simulations for mutually unbiased bases, Pauli observables, and Pauli basis measurements. 

\end{abstract}


\maketitle

\section{Introduction}

Quantum state tomography is the task of reconstructing a quantum state from experimental data.
Many methods have been proposed for this problem.
Maximum-likelihood estimation \cite{hradil_quantum_1997,rehacec_iterative_2001}
is a popular universal approach that produces point estimators,
but error bars are only available in asymptotic scenarios  
involving fully mixed states \footnote{Experimental works \cite{Haffner} do provide bootstrap error bars,
but these methods are not theoretically grounded for rank-deficient states. 
(General statistical theory provides asymptotic error bars and bootstrap theory, but only for fully mixed states where local asymptotic normality holds.)}.
This shortcoming spurred the development of alternatives, such as Bayesian \cite{blume-kohout_optimal_2010,GranadeCombes,GranadeFerrie} and region estimators \cite{christandl_reliable_2012,blume-kohout_robust_2012, faist_practical_2016}.  These methods have other drawbacks,
such as comparatively high computational cost and weak (or implicit) convergence guarantees. 
In parallel, researchers proposed compressed sensing techniques \cite{gross_quantum_2010,liu_universal_2011,flammia_quantum_2012, riofrio_experimental_2017}
to estimate (approximately) low-rank states from fewer samples.

In this work, we revisit linear inversion, one of the oldest and simplest methods for tomography {\bf \cite{DArianoPerinotti}}. 
We prove that a variant, \emph{projected least squares (PLS)},
supports easy-to-interpret, non-asymptotic error guarantees. 
Our main result demonstrates that roughly $r^2d \epsilon^{-2} \log d$ independent samples
suffice to reconstruct \emph{any} rank-$r$ state $\rho$ of level $d$ up to accuracy $\epsilon$ in trace distance. 
This sampling rate is competitive with the most powerful techniques in the literature \cite{flammia_quantum_2012,kueng_low_2017}, and it (almost) saturates fundamental lower bounds
on the minimal number of independent samples\footnote{Correlated measurements over several copies of the state can achieve even lower sampling rates \cite{haah_sample_2017,odonnel_efficient_2016}, but are much more challenging to implement.
Here we focus on the practically more relevant case of measuring each copy of $\rho$ independently.}
required for tomography \cite{haah_sample_2017,flammia_quantum_2012}.
As an added benefit, the PLS method is numerically ``cheap'' in the sense that its computational cost is dominated by forming the least-squares estimator. 
Numerical simulations indicate that PLS is much faster and performs well in comparison to prominent alternatives,
including maximum-likelihood estimation \cite{Haffner} and compressed sensing \cite{gross_quantum_2010}.

\subsection{Background and estimator}

For $d$-level systems, a tomographically complete measurement\footnote{A measurement is tomographically complete if, for every pair of distinct states $\rho \neq \sigma$, there exists $i \in \left[ m \right]$ such that $\mathrm{tr}(M_i \rho) \neq \mathrm{tr}(M_i \sigma)$.} is described by $d \times d$ hermitian matrices $M_1,\ldots,M_m \in \mathbb{H}_d$ that are positive semidefinite and obey $\sum_{i=1}^m M_i = \mathbb{I}$. Measuring a quantum state $\rho$ results in one of $m$ outcomes, indexed by $i \in \left[m \right]$. The probability of observing outcome $i$ depends on $\rho$ and is described by Born's rule:
\begin{equation}
\left[p \right]_i = \mathrm{Pr} \left[ i| \rho \right] = \mathrm{tr} \left( M_i \rho \right) \textrm{ for } i \in \left[m \right]. \label{eq:born_rule}
\end{equation}
These probabilities can be estimated by \emph{frequencies}:
Prepare $n$ copies of the state, measure each of them separately, and set
\begin{equation}
\left[ f_n \right]_i= \frac{n_i}{n} \textrm{ for } i \in \left[m \right], \label{eq:frequencies}
\end{equation}
where $n_i$ is the number of times outcome $i$ was observed.
These frequencies converge to the true probabilities as $n \to \infty$.
The \emph{least-squares estimator}
is the solution to the least-squares problem that results from replacing the true probabilities
in Born's rule by frequencies \eqref{eq:frequencies}:
\begin{equation}
\hat{L}_n = \underset{X \in \mathbb{H}_d}{\textrm{argmin}} \quad \sum_{i=1}^m \left( \left[f_n \right]_i - \mathrm{tr} \left( M_i X \right) \right)^2 . \label{eq:linear_inversion}
\end{equation}
This optimization inverts the $m$ linear equations specified by Born's rule; see Eq.~\eqref{eq:LI_solution} below (\emph{linear inversion}). In general, $\hat{L}_n$ can have negative eigenvalues, so it may \emph{fail} to be a quantum state.  Several ways to overcome this drawback have been proposed; e.g.\ \cite{kaznady_numerical_2009,smolin_efficient_2012,AlquierButucea_2013,butucea_spectral_2015,Koltchinskii_optimal_2015}.
In this work, we propose to compute the quantum state closest to $\hat{L}_n$ with respect to Frobenius norm:
\begin{equation}
\hat{\rho}_n = \underset{\sigma\textrm{ is a quantum state}}{\textrm{argmin}} \quad \| \hat{L}_n-\sigma \|_2.
\label{eq:sdp}
\end{equation}
Mathematically, this corresponds to \emph{projecting} $\hat{L}_n$ onto the convex set of all quantum states. 
We term this procedure \emph{Projected Least Squares (PLS)}; see Tab.~\ref{tab:box}.
To our knowledge, PLS is new.  Indeed, existing techniques are (apparently) more sophisticated.
Nevertheless, the simplicity of our approach is a key advantage for both the analysis and for the
actual computation.
The least-squares estimator \eqref{eq:linear_inversion}
and the projection onto the set of quantum states \eqref{eq:analytic_solution} admit closed-form expressions.
The total cost of the PLS estimator is dominated by forming the least-squares estimate $\hat{L}_n$.
PLS requires considerably
less storage and arithmetic than existing techniques that are based on more complicated
optimization problems.

\begin{table}
\begin{tcolorbox}[
left=0mm,right=0mm,top=0mm,bottom=0mm,boxsep=1mm,arc=0mm,boxrule=0.5pt,
title=Projected least squares (PLS) estimator $\hat{\rho}_n$]
\begin{itemize}
\item Estimate probabilities by frequencies \eqref{eq:frequencies},
\item Compute the least squares estimator \eqref{eq:linear_inversion},
\item Project it onto the set of quantum states \eqref{eq:sdp}.
\end{itemize}
\end{tcolorbox}
\caption{Summary of the estimation technique.} \label{tab:box}
\end{table}

\section{Results}

\subsection{Error bounds and confidence regions for $\hat{\rho}_n$} 

We analyze several important and practically relevant measurement systems: \emph{structured POVMs} (e.g. SIC-POVMs, MUBs and stabilizer states), \emph{(global) Pauli observables}, \emph{Pauli basis measurements}, and the \emph{uniform/covariant POVM}. 
For each of these settings, the PLS estimator $\hat{\rho}_n$ provably converges to the true state $\rho$
in trace distance $\| \cdot \|_1$.  

\begin{theorem}[Error bound for $\hat{\rho}_n$] \label{thm:main_result}
Let $\rho \in \mathbb{H}_d$ be state and fix a number of samples $n \in \mathbb{N}$.
Then, for \emph{each} of the aforementioned measurements, the PLS estimator (Tab.~\ref{tab:box}) obeys
\begin{equation*}
\mathrm{Pr} \left[ \left\| \hat{\rho}_n- \rho \right\|_1 \geq \epsilon \right] \leq d \mathrm{e}^{- \frac{n \epsilon^2}{43 g(d) r^2}} \textrm{ for } \epsilon \in \left[0,1 \right],
\end{equation*}
where $r = \min \left\{ \mathrm{rank}(\rho), \mathrm{rank}(\hat{\rho}_n) \right\}$ and $g(d)$ specifies dependence on the ambient dimension:
\begin{align*}
g(d)  &= 2 d \quad \textrm{for structured POVMs; see Eq.~\eqref{eq:LI_2design} for $\hat{L}_n$}, \\
g(d) &=  d^2 \quad \textrm{for Pauli observables; see Eq.~\eqref{eq:LI_pauli} for $\hat{L}_n$}, \\
g(d) &\simeq  d^{1.6} \quad \textrm{for Pauli basis measurements; see Eq.~\eqref{eq:LI_pauli_basis} for $\hat{L}_n$}.
\end{align*}
\end{theorem}

The following immediate Corollary endows $\hat{\rho}_n$ with rigorous error-bars in trace distance.

\begin{corollary}[$\delta$-confidence region] \label{cor:error_bars}
The trace-norm ball of size $ \mathrm{rank}(\hat{\rho}_n) \sqrt{43 g(d) n^{-1}\log\left(d/\delta\right)}$ around $\hat{\rho}_n$ (intersected with state space) is a $\delta$-confidence region for the true state $\rho$.
\end{corollary}

We emphasize the following aspects of this result:

\begin{itemize}
\item[(i)] \emph{(Almost) optimal sampling rate:}
Theorem~\ref{thm:main_result} highlights that 
\begin{equation}
n \geq 43 g(d) \frac{\mathrm{rank}(\rho)^2}{\epsilon^2} \log \left( \frac{d}{\delta} \right)
\label{eq:sampling_rate}
\end{equation}
samples suffice to ensure $\| \hat{\rho}_n - \rho \|_1 \leq \epsilon$ with probability at least $1-\delta$. 
For structured POVMs and Pauli observables, this \emph{sampling rate} is comparable to the
\emph{best} theoretical bounds for alternative tomography algorithms
\cite{kueng_clifford_2016,gross_quantum_2010}. 
Moreover, fundamental lower bounds in \cite{haah_sample_2017} and \cite{flammia_quantum_2012} indicate that this scaling is optimal up to a single $\log (d)$-factor, so it cannot be improved substantially.

\item[(ii)] \emph{implicit exploitation of (approximate) low rank:} the number of samples required to achieve a good estimator scales quadratically in the rank, rather than the ambient dimension $d$.  This behavior extends to the case where $\rho$, or $\hat{\rho}_n$, is well-approximated by a rank-$r$ matrix; see Theorem~\ref{thm:main_appendix} in the appendix.  These results are comparable with guarantees for compressed sensing methods \cite{flammia_quantum_2012} that are specifically designed to exploit low-rank. Fig~\ref{fig:cs} (below) provides numerical confirmation.
\end{itemize}

\begin{proof}[Proof sketch for Theorem~\ref{thm:main_result}]
The least-squares estimator $\hat{L}_n$ can be viewed as a sum of $n$ independent random matrices.
To illustrate this, consider a single structured POVM measurement. Then $\hat{L}_1$ defined in \eqref{eq:LI_2design} is an instance of the random matrix $X=|v_k \rangle \! \langle v_k|-\mathbb{I}$, where $k \in \left[m \right]$ occurs with probability $\langle v_k| \rho |v_k \rangle$ (Born's rule). 
This generalizes to $\hat{L}_n = \frac{1}{n} \sum_{i=1}^n X_i$, where the matrices $X_i$ are statistically independent. Such sums of random matrices concentrate sharply around their expectation value $\mathbb{E} \hat{L}_n = \rho$, and 
matrix concentration inequalities \cite{tropp_user-friendly_2012} quantify this convergence: 
\begin{equation}
\mathrm{Pr} \left[ \left\| \hat{L}_n - \rho \right\|_\infty \geq \tau \right] \leq d \mathrm{e}^{ - \frac{3n \tau^2}{8g(d)} }\quad \tau \in \left[0,1 \right].
\label{eq:essential}
\end{equation}
The operator norm bound induces a Frobenius norm bound, and the projection onto quantum states contracts the Frobenius norm. 
The claim then follows from relating the (projected) Frobenius norm distance to the trace distance. None of these comparisons depend on the ambient dimension, but only on the (approximate) rank. We refer to the appendix for details.
\end{proof}

\subsection{Optimal performance guarantee for the uniform POVM}

Theorem~\ref{thm:main_result} involves a factor of the dimension $d$ that may be extraneous.
In turn, this factor introduces an additional $\log (d)$-gap between Eq.~\eqref{eq:sampling_rate}
and existing lower bounds \cite{haah_sample_2017}.
The dimensional factors emerge because we employ matrix-valued concentration inequalities in the proof.
Our second main result shows that we can remove the dimensional factor for the \emph{uniform POVM},
which encompasses all rank-one projectors:

\begin{theorem}[Convergence of $\hat{\rho}_n$ for the uniform POVM] \label{thm:uniform_main}
For uniform POVM measurements, the PLS estimator obeys
\begin{equation*}
\mathrm{Pr} \left[ \left\| \hat{\rho}_n-\rho \right\|_1 \geq \epsilon \right] \leq \mathrm{e}^{ 2.2 d - \frac{\epsilon^2 n}{480 r^2} } \textrm{ for }  \epsilon >0,
\end{equation*}
where $r = \min \left\{ \mathrm{rank}(\rho),\mathrm{rank}(\hat{\rho}_n)\right\}$.
\end{theorem}

This result exactly reproduces the best existing performance guarantees
for tomography from independent measurements \cite{kueng_low_2017}. 
The bound follows from standard techniques from high-dimensional probability theory.

\begin{proof}[Proof sketch of Theorem~\ref{thm:uniform_main}]
The operator norm  has a variational formulation: $\| \hat{L}_n - \rho \|_\infty = \max_{z \in \mathbb{S}^d} |\langle z| \hat{L}_n - \rho |z \rangle$. The optimization over the unit sphere may be replaced by a maximization over a finite point set, called a \emph{covering net}, whose cardinality scales exponentially in $d$. 
For any $z \in \mathbb{S}^d$, $\langle z| \hat{L}_n-\rho|z \rangle$ is a sum of $n$ i.i.d. random variables that exhibit subexponential tail decay. (Measuring the uniform POVM allows us to draw this conclusion.) Standard concentration inequalities yield a tail bound that decays exponentially in the number $n$ of samples.
Applying a union bound over all points $z_i$ in the net then ensures
$
\mathrm{Pr} \left[ \| \hat{L}_n - \rho \|_\infty \geq \tau \right]
\leq 2 \mathrm{e}^{ c_1 d - c_2 n \tau^2}.
$
Subsequently, closeness in operator norm for $\hat{L}_n$ may be converted into closeness in trace-norm for $\hat{\rho}_n$ at the cost of an additional (effective) rank factor.
\end{proof}

\section{Algorithmic considerations}

\subsection{Explicit solutions for the least squares estimator \eqref{eq:linear_inversion}}

Tomographically complete measurements can be viewed as injective linear maps $\mathcal{M}: \mathbb{H}_d \to \mathbb{R}^m$ with components $\left[ \mathcal{M}(X) \right]_i = \mathrm{tr}(M_i X)$
 specified by Born's rule \eqref{eq:born_rule}. It is well known that the  least-squares problem \eqref{eq:linear_inversion} admits the closed-form solution:
\begin{equation}
\hat{L}_n = \left( \mathcal{M}^\dagger \mathcal{M} \right)^{-1} \left( \mathcal{M}^\dagger \!(f_n) \right). \label{eq:LI_solution}
\end{equation}
We evaluate this formula for different measurements
and content ourselves with sketching key steps and results (see appendix for details).

\paragraph{\textbf{Structured POVMs and the uniform POVM:}}
Also known as 2-designs,
these systems include highly structured, rank-one POVMs $\left\{ \frac{d}{m} |v_i \rangle \! \langle v_i| \right\}_{i=1}^m$, such as
symmetric informationally complete  POVMs \cite{renes_symmetric_2004}, maximal sets of mutually unbiased bases \cite{klappenecker_mutually_2005}, the set of all stabilizer states \cite{dankert_exact_2009,gross_evenly_2007}, as well as the uniform POVM.
By definition, for $X \in \mathbb{H}_d$, all of the above systems obey
\begin{align*}
\mathcal{M}^\dagger \!\mathcal{M}(X)
= \frac{d^2}{m} \sum_{i=1}^m \langle v_i| X |v_i \rangle \!|v_i \rangle \! \langle v_i | = \frac{md}{d+1}\left( X+ \mathrm{tr}(X) \mathbb{I} \right).
\end{align*}
These equations can readily be inverted, and Eq.~\eqref{eq:LI_solution} simplifies to
\begin{equation}
\hat{L}_n = (d+1) \sum_{i=1}^m  \left[ f_n \right]_i  |v_i \rangle \! \langle v_i| - \mathbb{I}. \label{eq:LI_2design}
\end{equation}

\paragraph{\textbf{Pauli observables:}}

Fix $d=2^k$ ($k$ qubits), and let $W_1,\ldots,W_{d^2} \in \mathbb{H}_d$ be the set of Pauli observables, comprising all possible $k$-fold tensor products of the elementary $2 \times 2$ Pauli matrices.
We can approximate the expectation value $\mathrm{tr}(W_i \rho)$ of each Pauli observable by the empirical mean 
$\hat{\mu}_i = \left[ f^+_{n/d^2}\right]_i - \left[f^-_{n/d^2} \right]$
of the 2-outcome POVM $P_i^\pm =  \frac{1}{2} \left( \mathbb{I} \pm W_i \right) $.
Pauli matrices form a unitary operator basis, and the evaluation of Eq.~\eqref{eq:LI_solution} is simple:
\begin{equation}
\hat{L}_n = \frac{1}{d} \sum_{i=1}^{d^2} \hat{\mu}_i W_i
= \frac{1}{d} \sum_{i=1}^{d^2} \left( \left[ f_{n/d^2}^+ \right]_i - \left[ f_{n/d^2}^- \right]_i \right) W_i. \label{eq:LI_pauli}
\end{equation}

\paragraph{\textbf{Pauli basis measurements}}

Rather than approximating (global) expectation values, it is possible to perform different combinations of local Pauli measurements.
For $d=2^k$, there are $3^k$ potential combinations in total.
Each of the settings $\vec{s} \in \left\{x,y,z \right\}^k$ corresponds to a basis measurement  $ | b^{(\vec{s})} _{\vec{o}} \rangle \! \langle b^{(\vec{s})}_{\vec{o}}|$, where $\vec{o} \in\left\{ \pm 1 \right\}^k$ labels the $2^k$ potential outcomes. 
The union $\mathcal{M}$ of all $3^k$ bases obeys $\left( \mathcal{M}^\dagger \mathcal{M} \right)(X) = 3^k \mathcal{D}^{\otimes k}_{1/3}(X)$, 
where $\mathcal{D}_{1/3}(X) =  \frac{1}{3} \rho + \frac{\mathrm{tr}(X)}{3} \mathbb{I}$ denotes a single-qubit depolarizing channel. Evaluating Eq.~\eqref{eq:LI_solution} yields
\begin{equation}
\hat{L}_n
= \frac{1}{3^k} \sum_{\vec{s},\vec{o}} \left[ f_{n/3^k} \right]_{\vec{o}}^{(\vec{s})} \left( \mathcal{D}_{1/3}^{\otimes k} \right)^{-1}\left( | b^{(\vec{s})} _{\vec{o}} \rangle \! \langle b^{(\vec{s})} _{\vec{o}}|  \right). \label{eq:LI_pauli_basis}
\end{equation}

Finally, we point out that all of these explicit solutions are guaranteed to have unit trace: $\mathrm{tr} \left( \hat{L}_n \right) = 1$.  They are generally not positive semidefinite,

\subsection{Explicit solutions for the projection step \eqref{eq:sdp}}

The PLS estimator is defined to be the state closest in Frobenius norm to the least-squares estimator $\hat{L}_n$. 
The search \eqref{eq:sdp} admits a simple, \emph{analytic} solution \cite{smolin_efficient_2012}. 
 Define the all-ones vector $\vec{1} \in \mathbb{R}^d$ and the thresholding function $\left[ \cdot \right]^+$ with components $\left[ \vec{y} \right]^+_i = \max \left\{ \left[\vec{y}\right]_i, 0 \right\}$. Let $\hat{L}_n = U \mathrm{diag}(\vec{\lambda}) U^\dagger$ be an eigenvalue decomposition.
Then
\begin{equation}
\hat{\rho}_n = U \mathrm{diag} \left(\left[ \vec{\lambda} - x_0 \vec{1} \right]^+\right) U^\dagger, \label{eq:analytic_solution}
\end{equation}
where $x_0 \in \mathbb{R}$ is chosen so that $\mathrm{tr}(\hat{\rho}_n) = 1$. The fact that $\hat{L}_n$ itself has unit trace ensures that this solution to Eq.~\eqref{eq:sdp} is unique. The number $x_0$ may be determined by applying a root-finding algorithm to the non-increasing function $f(x) = 2+\mathrm{tr}(\hat{L}_n) - dx + \sum_{i=1}^d \left| \lambda_i - x \right|$. 

\subsection{Runtime analysis}

The two steps discussed here are inherently \emph{scalable}: just count frequencies to determine the LI estimators (\ref{eq:LI_2design},\ref{eq:LI_pauli},\ref{eq:LI_pauli_basis}) at a total cost of (at most) $\ \min \left\{m,n \right\}$ matrix additions. The subsequent projection onto state-space is a particular type of soft-thresholding. The associated computational cost is dominated by the eigenvalue decomposition and has runtime (at most) $\mathcal{O} \left(d^3 \right)$.

In summary, forming $\hat{L}_n$ is the dominant cost of a na{\"i}ve implementation.
However, the high degree of structure may allow us to employ techniques from randomized
linear algebra~\cite{halko_finding_2011} to reduce the cost.

\section{Numerical experiments}

We numerically compare the performance of PLS to \emph{maximum likelihood} (ML) and \emph{compressed sensing} (CS), respectively. Additional numerical studies for structured POVMs can be found in the appendix.

Fig.~\ref{fig:mlvsPLI_varying_repetitions} compares ML and PLS for Pauli basis measurements in dimension $d=2^4$. 
The trace-norm error incurred by PLS is within a factor of two of ML for low-rank states. Additional simulations (not included here) indicate that this gap closes for full-rank states.

\begin{figure} 
\includegraphics[width=0.44\textwidth]{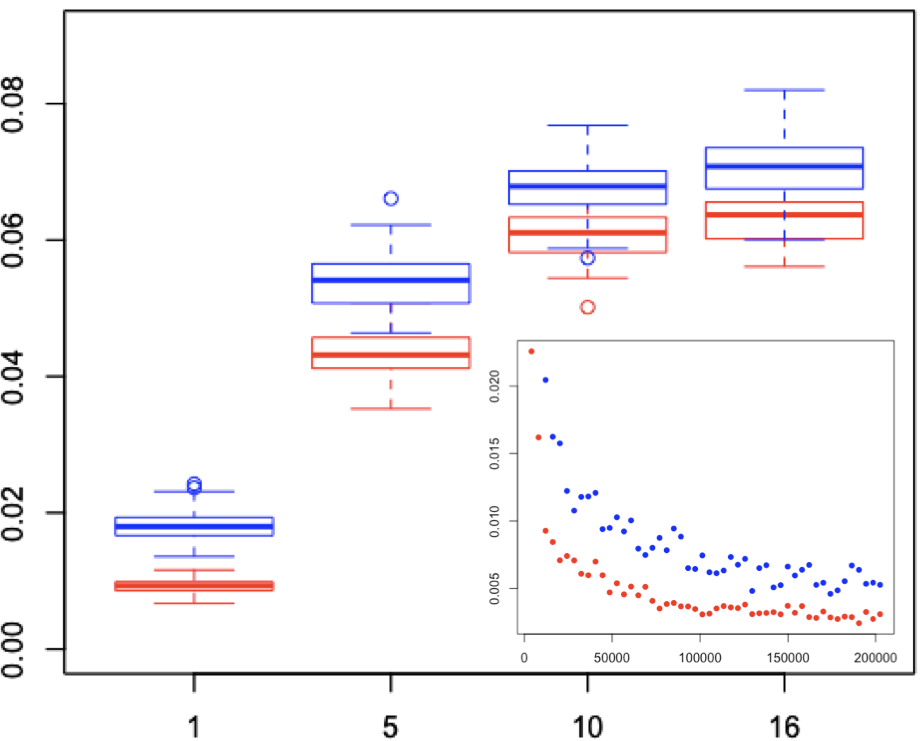}
\caption{PLS (blue) vs.\ ML (red) for 4-qubit Pauli basis measurements: boxplots of trace distance error for ML vs.\ PLS for 100 datasets generated with random states of rank 1,5,10 and 16, and 200 repetitions per setting.  \emph{Inset:} trace distance error as a function of sample size for a pure target state.} \label{fig:mlvsPLI_varying_repetitions}
\end{figure}

CS is a natural benchmark for low-rank tomography.  The papers \cite{gross_quantum_2010,liu_universal_2011} apply to Pauli observables, and they show that a random choice of $m\geq C rd \log^6 (d)$ Pauli observables is sufficient to reconstruct \emph{any} rank-$r$ state. 
The actual reconstruction is performed by solving a convex optimization problem, namely the least-squares fit over the set of quantum states \cite{kalev_quantum_2015,kabanava_stable_2016}. 
Numerical studies from \cite{flammia_quantum_2012} suggest that $m=256$ is appropriate for $d=2^5$ and $r=1$. 
Fig.~\ref{fig:cs} shows that PLS consistently outperforms the CS estimator in this regime.
Importantly, PLS was also much faster to evaluate than both, ML and CS.

\begin{figure} 
\includegraphics[width=0.45\textwidth]{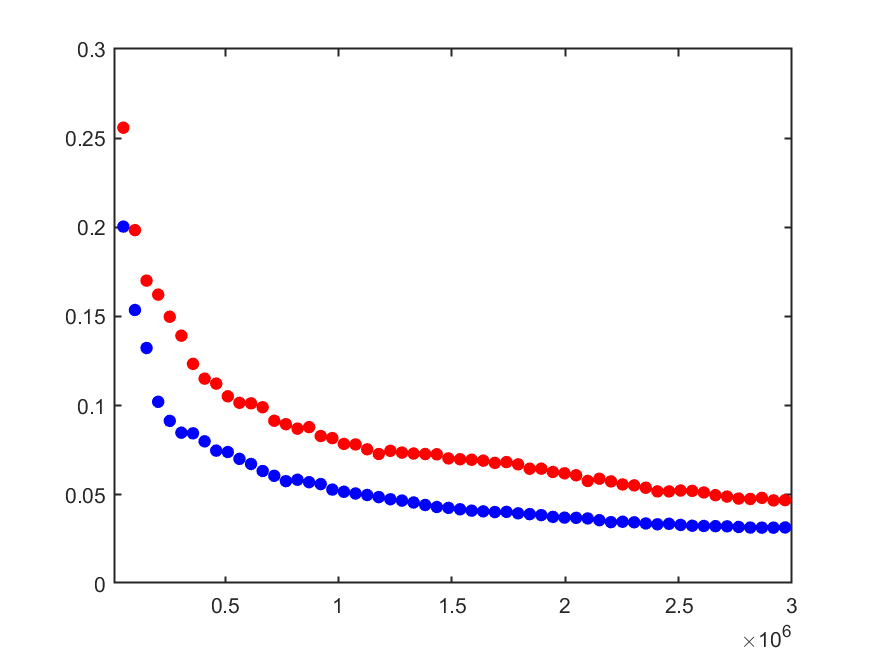}
\caption{PLS (blue) vs.\ CS (red) for $m$ 5-qubit Pauli observables and a pure target state: Trace distance error for CS ($m=256$)  and PLS ($m=1024$) as a function of (total) sample size.} \label{fig:cs}
\end{figure}

\section{Conclusion and Outlook}

Linear inversion is one of the oldest and simplest approaches to solve the practically important task of quantum state tomography.
In this work, we introduced a variant called \emph{projected least squares} (PLS)
that projects the least-squares estimator onto the set of all quantum states.
Not only is this estimator numerically cheap, but it comes with
strong, non-asymptotic convergence guarantees.  These results are derived using
concentration inequalities for sums of random matrices, and they exploit the
randomness inherent in quantum experiments. 

We show that PLS is competitive, both in theory and in practice.  For a variety of measurements, the results
match the \emph{best} existing theoretical results for the sampling rate of other tomography methods. 
In particular, for the uniform POVM, an order of $\frac{r d^2}{\epsilon^2}$ samples suffice to reconstruct any rank-$r$ state up to accuracy $\epsilon$ in trace distance.  This result also saturates existing lower bounds \cite{haah_sample_2017} on the minimal sampling rate required for \emph{any} tomographic procedure with independent measurements.
Numerical studies underline these competitive features.

\paragraph*{\textbf{Outlook:}} 
Corollary~\ref{cor:error_bars} is not (yet) optimal. \emph{Bootstrapping} could be used to obtain tighter confidence regions, and the low computational cost of PLS may speed up this process considerably. 
It also seems fruitful to combine the ideas presented here with recent insights from \cite{scholten_behavior_2018}.
Finally, the proof of Theorem~\ref{thm:main_result} indicates that PLS is stable with respect to time-dependent state generation (\emph{drift}). 
We intend to address these points in future work.

\paragraph*{\textbf{Acknowledgements:}}

The authors thank  Philippe Faist, Anirudh Acharya and Theodore Kypraios for fruitful discussions and valuable feedback.
RK and JT are supported by ONR Award No.\ N00014-17-12146. RK also acknowledges funding provided by the Institute of Quantum Information
and Matter, an NSF Physics Frontiers Center (NSF Grant PHY-1733907).

\cleardoublepage

\bibliographystyle{apsrev4-1}
\bibliography{tomography}

\begin{thebibliography}{54}%
\makeatletter
\providecommand \@ifxundefined [1]{%
 \@ifx{#1\undefined}
}%
\providecommand \@ifnum [1]{%
 \ifnum #1\expandafter \@firstoftwo
 \else \expandafter \@secondoftwo
 \fi
}%
\providecommand \@ifx [1]{%
 \ifx #1\expandafter \@firstoftwo
 \else \expandafter \@secondoftwo
 \fi
}%
\providecommand \natexlab [1]{#1}%
\providecommand \enquote  [1]{``#1''}%
\providecommand \bibnamefont  [1]{#1}%
\providecommand \bibfnamefont [1]{#1}%
\providecommand \citenamefont [1]{#1}%
\providecommand \href@noop [0]{\@secondoftwo}%
\providecommand \href [0]{\begingroup \@sanitize@url \@href}%
\providecommand \@href[1]{\@@startlink{#1}\@@href}%
\providecommand \@@href[1]{\endgroup#1\@@endlink}%
\providecommand \@sanitize@url [0]{\catcode `\\12\catcode `\$12\catcode
  `\&12\catcode `\#12\catcode `\^12\catcode `\_12\catcode `\%12\relax}%
\providecommand \@@startlink[1]{}%
\providecommand \@@endlink[0]{}%
\providecommand \url  [0]{\begingroup\@sanitize@url \@url }%
\providecommand \@url [1]{\endgroup\@href {#1}{\urlprefix }}%
\providecommand \urlprefix  [0]{URL }%
\providecommand \Eprint [0]{\href }%
\providecommand \doibase [0]{http://dx.doi.org/}%
\providecommand \selectlanguage [0]{\@gobble}%
\providecommand \bibinfo  [0]{\@secondoftwo}%
\providecommand \bibfield  [0]{\@secondoftwo}%
\providecommand \translation [1]{[#1]}%
\providecommand \BibitemOpen [0]{}%
\providecommand \bibitemStop [0]{}%
\providecommand \bibitemNoStop [0]{.\EOS\space}%
\providecommand \EOS [0]{\spacefactor3000\relax}%
\providecommand \BibitemShut  [1]{\csname bibitem#1\endcsname}%
\let\auto@bib@innerbib\@empty
\bibitem [{\citenamefont {Hradil}(1997)}]{hradil_quantum_1997}%
  \BibitemOpen
  \bibfield  {author} {\bibinfo {author} {\bibfnamefont {Z.}~\bibnamefont
  {Hradil}},\ }\href {\doibase 10.1103/PhysRevA.55.R1561} {\bibfield  {journal}
  {\bibinfo  {journal} {Phys. Rev. A}\ }\textbf {\bibinfo {volume} {55}},\
  \bibinfo {pages} {R1561} (\bibinfo {year} {1997})}\BibitemShut {NoStop}%
\bibitem [{\citenamefont {\ifmmode \check{R}\else
  \v{R}\fi{}eh\'a\ifmmode~\check{c}\else \v{c}\fi{}ek}\ \emph
  {et~al.}(2001)\citenamefont {\ifmmode \check{R}\else
  \v{R}\fi{}eh\'a\ifmmode~\check{c}\else \v{c}\fi{}ek}, \citenamefont
  {Hradil},\ and\ \citenamefont {Je\ifmmode~\check{z}\else
  \v{z}\fi{}ek}}]{rehacec_iterative_2001}%
  \BibitemOpen
  \bibfield  {author} {\bibinfo {author} {\bibfnamefont {J.}~\bibnamefont
  {\ifmmode \check{R}\else \v{R}\fi{}eh\'a\ifmmode~\check{c}\else
  \v{c}\fi{}ek}}, \bibinfo {author} {\bibfnamefont {Z.}~\bibnamefont {Hradil}},
  \ and\ \bibinfo {author} {\bibfnamefont {M.}~\bibnamefont
  {Je\ifmmode~\check{z}\else \v{z}\fi{}ek}},\ }\href {\doibase
  10.1103/PhysRevA.63.040303} {\bibfield  {journal} {\bibinfo  {journal} {Phys.
  Rev. A}\ }\textbf {\bibinfo {volume} {63}},\ \bibinfo {pages} {040303}
  (\bibinfo {year} {2001})}\BibitemShut {NoStop}%
\bibitem [{\citenamefont {Blume-Kohout}(2010)}]{blume-kohout_optimal_2010}%
  \BibitemOpen
  \bibfield  {author} {\bibinfo {author} {\bibfnamefont {R.}~\bibnamefont
  {Blume-Kohout}},\ }\href {http://stacks.iop.org/1367-2630/12/i=4/a=043034}
  {\bibfield  {journal} {\bibinfo  {journal} {New J. Phys.}\ }\textbf {\bibinfo
  {volume} {12}},\ \bibinfo {pages} {043034} (\bibinfo {year}
  {2010})}\BibitemShut {NoStop}%
\bibitem [{\citenamefont {Granade}\ \emph {et~al.}(2016)\citenamefont
  {Granade}, \citenamefont {Combes},\ and\ \citenamefont
  {Corrie}}]{GranadeCombes}%
  \BibitemOpen
  \bibfield  {author} {\bibinfo {author} {\bibfnamefont {C.}~\bibnamefont
  {Granade}}, \bibinfo {author} {\bibfnamefont {J.}~\bibnamefont {Combes}}, \
  and\ \bibinfo {author} {\bibfnamefont {D.~G.}\ \bibnamefont {Corrie}},\
  }\href {http://stacks.iop.org/1367-2630/18/i=3/a=033024} {\bibfield
  {journal} {\bibinfo  {journal} {New J. Phys}\ }\textbf {\bibinfo {volume}
  {18}},\ \bibinfo {pages} {033024} (\bibinfo {year} {2016})}\BibitemShut
  {NoStop}%
\bibitem [{\citenamefont {Granade}\ \emph {et~al.}(2017)\citenamefont
  {Granade}, \citenamefont {Ferrie},\ and\ \citenamefont
  {Flammia}}]{GranadeFerrie}%
  \BibitemOpen
  \bibfield  {author} {\bibinfo {author} {\bibfnamefont {C.}~\bibnamefont
  {Granade}}, \bibinfo {author} {\bibfnamefont {C.}~\bibnamefont {Ferrie}}, \
  and\ \bibinfo {author} {\bibfnamefont {S.~T.}\ \bibnamefont {Flammia}},\
  }\href {http://stacks.iop.org/1367-2630/19/i=11/a=113017} {\bibfield
  {journal} {\bibinfo  {journal} {New J. Phys.}\ }\textbf {\bibinfo {volume}
  {19}},\ \bibinfo {pages} {113017} (\bibinfo {year} {2017})}\BibitemShut
  {NoStop}%
\bibitem [{\citenamefont {Christandl}\ and\ \citenamefont
  {Renner}(2012)}]{christandl_reliable_2012}%
  \BibitemOpen
  \bibfield  {author} {\bibinfo {author} {\bibfnamefont {M.}~\bibnamefont
  {Christandl}}\ and\ \bibinfo {author} {\bibfnamefont {R.}~\bibnamefont
  {Renner}},\ }\href {\doibase 10.1103/PhysRevLett.109.120403} {\bibfield
  {journal} {\bibinfo  {journal} {Phys. Rev. Lett.}\ }\textbf {\bibinfo
  {volume} {109}},\ \bibinfo {pages} {120403} (\bibinfo {year}
  {2012})}\BibitemShut {NoStop}%
\bibitem [{\citenamefont {Blume-Kohout}(2012)}]{blume-kohout_robust_2012}%
  \BibitemOpen
  \bibfield  {author} {\bibinfo {author} {\bibfnamefont {R.}~\bibnamefont
  {Blume-Kohout}},\ }\href {https://arxiv.org/abs/1202.5270} {\bibfield
  {journal} {\bibinfo  {journal} {arXiv preprint arXiv:1202.5270}\ } (\bibinfo
  {year} {2012})}\BibitemShut {NoStop}%
\bibitem [{\citenamefont {Faist}\ and\ \citenamefont
  {Renner}(2016)}]{faist_practical_2016}%
  \BibitemOpen
  \bibfield  {author} {\bibinfo {author} {\bibfnamefont {P.}~\bibnamefont
  {Faist}}\ and\ \bibinfo {author} {\bibfnamefont {R.}~\bibnamefont {Renner}},\
  }\href {\doibase 10.1103/PhysRevLett.117.010404} {\bibfield  {journal}
  {\bibinfo  {journal} {Phys. Rev. Lett.}\ }\textbf {\bibinfo {volume} {117}},\
  \bibinfo {pages} {010404} (\bibinfo {year} {2016})}\BibitemShut {NoStop}%
\bibitem [{\citenamefont {Gross}\ \emph {et~al.}(2010)\citenamefont {Gross},
  \citenamefont {Liu}, \citenamefont {Flammia}, \citenamefont {Becker},\ and\
  \citenamefont {Eisert}}]{gross_quantum_2010}%
  \BibitemOpen
  \bibfield  {author} {\bibinfo {author} {\bibfnamefont {D.}~\bibnamefont
  {Gross}}, \bibinfo {author} {\bibfnamefont {Y.-K.}\ \bibnamefont {Liu}},
  \bibinfo {author} {\bibfnamefont {S.~T.}\ \bibnamefont {Flammia}}, \bibinfo
  {author} {\bibfnamefont {S.}~\bibnamefont {Becker}}, \ and\ \bibinfo {author}
  {\bibfnamefont {J.}~\bibnamefont {Eisert}},\ }\href {\doibase
  10.1103/PhysRevLett.105.150401} {\bibfield  {journal} {\bibinfo  {journal}
  {Phys. Rev. Lett.}\ }\textbf {\bibinfo {volume} {105}},\ \bibinfo {pages}
  {150401} (\bibinfo {year} {2010})}\BibitemShut {NoStop}%
\bibitem [{\citenamefont {Liu}(2011)}]{liu_universal_2011}%
  \BibitemOpen
  \bibfield  {author} {\bibinfo {author} {\bibfnamefont {Y.-K.}\ \bibnamefont
  {Liu}},\ }in\ \href
  {http://papers.nips.cc/paper/4222-universal-low-rank-matrix-recovery-from-pauli-measurements.pdf}
  {\emph {\bibinfo {booktitle} {Adv. Neural Inf. Process. Syst. 24}}},\
  \bibinfo {editor} {edited by\ \bibinfo {editor} {\bibfnamefont
  {J.}~\bibnamefont {Shawe-Taylor}}, \bibinfo {editor} {\bibfnamefont {R.~S.}\
  \bibnamefont {Zemel}}, \bibinfo {editor} {\bibfnamefont {P.~L.}\ \bibnamefont
  {Bartlett}}, \bibinfo {editor} {\bibfnamefont {F.}~\bibnamefont {Pereira}}, \
  and\ \bibinfo {editor} {\bibfnamefont {K.~Q.}\ \bibnamefont {Weinberger}}}\
  (\bibinfo  {publisher} {Curran Associates, Inc.},\ \bibinfo {year} {2011})\
  pp.\ \bibinfo {pages} {1638--1646}\BibitemShut {NoStop}%
\bibitem [{\citenamefont {Flammia}\ \emph {et~al.}(2012)\citenamefont
  {Flammia}, \citenamefont {Gross}, \citenamefont {Liu},\ and\ \citenamefont
  {Eisert}}]{flammia_quantum_2012}%
  \BibitemOpen
  \bibfield  {author} {\bibinfo {author} {\bibfnamefont {S.~T.}\ \bibnamefont
  {Flammia}}, \bibinfo {author} {\bibfnamefont {D.}~\bibnamefont {Gross}},
  \bibinfo {author} {\bibfnamefont {Y.-K.}\ \bibnamefont {Liu}}, \ and\
  \bibinfo {author} {\bibfnamefont {J.}~\bibnamefont {Eisert}},\ }\href
  {http://stacks.iop.org/1367-2630/14/i=9/a=095022} {\bibfield  {journal}
  {\bibinfo  {journal} {New J. Phys.}\ }\textbf {\bibinfo {volume} {14}},\
  \bibinfo {pages} {095022} (\bibinfo {year} {2012})}\BibitemShut {NoStop}%
\bibitem [{\citenamefont {Riofr{\'\i}o}\ \emph {et~al.}(2017)\citenamefont
  {Riofr{\'\i}o}, \citenamefont {Gross}, \citenamefont {Flammia}, \citenamefont
  {Monz}, \citenamefont {Nigg}, \citenamefont {Blatt},\ and\ \citenamefont
  {Eisert}}]{riofrio_experimental_2017}%
  \BibitemOpen
  \bibfield  {author} {\bibinfo {author} {\bibfnamefont {C.}~\bibnamefont
  {Riofr{\'\i}o}}, \bibinfo {author} {\bibfnamefont {D.}~\bibnamefont {Gross}},
  \bibinfo {author} {\bibfnamefont {S.}~\bibnamefont {Flammia}}, \bibinfo
  {author} {\bibfnamefont {T.}~\bibnamefont {Monz}}, \bibinfo {author}
  {\bibfnamefont {D.}~\bibnamefont {Nigg}}, \bibinfo {author} {\bibfnamefont
  {R.}~\bibnamefont {Blatt}}, \ and\ \bibinfo {author} {\bibfnamefont
  {J.}~\bibnamefont {Eisert}},\ }\href
  {https://www.nature.com/articles/ncomms15305} {\bibfield  {journal} {\bibinfo
   {journal} {Nat. Commun.}\ }\textbf {\bibinfo {volume} {8}},\ \bibinfo
  {pages} {15305} (\bibinfo {year} {2017})}\BibitemShut {NoStop}%
\bibitem [{\citenamefont {D'Ariano}\ and\ \citenamefont
  {Perinotti}(2007)}]{DArianoPerinotti}%
  \BibitemOpen
  \bibfield  {author} {\bibinfo {author} {\bibfnamefont {G.~M.}\ \bibnamefont
  {D'Ariano}}\ and\ \bibinfo {author} {\bibfnamefont {P.}~\bibnamefont
  {Perinotti}},\ }\href
  {https://link.aps.org/doi/10.1103/PhysRevLett.98.020403} {\bibfield
  {journal} {\bibinfo  {journal} {Phys. Rev. Lett.}\ }\textbf {\bibinfo
  {volume} {98}},\ \bibinfo {pages} {020403} (\bibinfo {year}
  {2007})}\BibitemShut {NoStop}%
\bibitem [{\citenamefont {Kueng}\ \emph {et~al.}(2017)\citenamefont {Kueng},
  \citenamefont {Rauhut},\ and\ \citenamefont {Terstiege}}]{kueng_low_2017}%
  \BibitemOpen
  \bibfield  {author} {\bibinfo {author} {\bibfnamefont {R.}~\bibnamefont
  {Kueng}}, \bibinfo {author} {\bibfnamefont {H.}~\bibnamefont {Rauhut}}, \
  and\ \bibinfo {author} {\bibfnamefont {U.}~\bibnamefont {Terstiege}},\ }\href
  {\doibase https://doi.org/10.1016/j.acha.2015.07.007} {\bibfield  {journal}
  {\bibinfo  {journal} {Appl. Comput. Harmon. Anal.}\ }\textbf {\bibinfo
  {volume} {42}},\ \bibinfo {pages} {88 } (\bibinfo {year} {2017})}\BibitemShut
  {NoStop}%
\bibitem [{\citenamefont {Haah}\ \emph {et~al.}(2017)\citenamefont {Haah},
  \citenamefont {Harrow}, \citenamefont {Ji}, \citenamefont {Wu},\ and\
  \citenamefont {Yu}}]{haah_sample_2017}%
  \BibitemOpen
  \bibfield  {author} {\bibinfo {author} {\bibfnamefont {J.}~\bibnamefont
  {Haah}}, \bibinfo {author} {\bibfnamefont {A.~W.}\ \bibnamefont {Harrow}},
  \bibinfo {author} {\bibfnamefont {Z.}~\bibnamefont {Ji}}, \bibinfo {author}
  {\bibfnamefont {X.}~\bibnamefont {Wu}}, \ and\ \bibinfo {author}
  {\bibfnamefont {N.}~\bibnamefont {Yu}},\ }\href {\doibase
  10.1109/TIT.2017.2719044} {\bibfield  {journal} {\bibinfo  {journal} {IEEE
  Trans. Inform. Theory}\ }\textbf {\bibinfo {volume} {63}},\ \bibinfo {pages}
  {5628} (\bibinfo {year} {2017})}\BibitemShut {NoStop}%
\bibitem [{\citenamefont {H{\"a}ffner}\ \emph {et~al.}(2005)\citenamefont
  {H{\"a}ffner}, \citenamefont {H{\"a}nsel}, \citenamefont {Roos},
  \citenamefont {Benhelm}, \citenamefont {Chek-al kar}, \citenamefont
  {Chwalla}, \citenamefont {K{\"o}rber}, \citenamefont {Rapol}, \citenamefont
  {Riebe}, \citenamefont {Schmidt}, \citenamefont {Becher}, \citenamefont
  {G{\"u}hne}, \citenamefont {D{\"u}r},\ and\ \citenamefont {Blatt}}]{Haffner}%
  \BibitemOpen
  \bibfield  {author} {\bibinfo {author} {\bibfnamefont {H.}~\bibnamefont
  {H{\"a}ffner}}, \bibinfo {author} {\bibfnamefont {W.}~\bibnamefont
  {H{\"a}nsel}}, \bibinfo {author} {\bibfnamefont {C.~F.}\ \bibnamefont
  {Roos}}, \bibinfo {author} {\bibfnamefont {J.}~\bibnamefont {Benhelm}},
  \bibinfo {author} {\bibfnamefont {D.}~\bibnamefont {Chek-al kar}}, \bibinfo
  {author} {\bibfnamefont {M.}~\bibnamefont {Chwalla}}, \bibinfo {author}
  {\bibfnamefont {T.}~\bibnamefont {K{\"o}rber}}, \bibinfo {author}
  {\bibfnamefont {U.~D.}\ \bibnamefont {Rapol}}, \bibinfo {author}
  {\bibfnamefont {M.}~\bibnamefont {Riebe}}, \bibinfo {author} {\bibfnamefont
  {P.~O.}\ \bibnamefont {Schmidt}}, \bibinfo {author} {\bibfnamefont
  {C.}~\bibnamefont {Becher}}, \bibinfo {author} {\bibfnamefont
  {O.}~\bibnamefont {G{\"u}hne}}, \bibinfo {author} {\bibfnamefont
  {W.}~\bibnamefont {D{\"u}r}}, \ and\ \bibinfo {author} {\bibfnamefont
  {R.}~\bibnamefont {Blatt}},\ }\href
  {https://www.nature.com/articles/nature04279} {\bibfield  {journal} {\bibinfo
   {journal} {Nature}\ }\textbf {\bibinfo {volume} {438}},\ \bibinfo {pages}
  {643 EP } (\bibinfo {year} {2005})}\BibitemShut {NoStop}%
\bibitem [{\citenamefont {Kaznady}\ and\ \citenamefont
  {James}(2009)}]{kaznady_numerical_2009}%
  \BibitemOpen
  \bibfield  {author} {\bibinfo {author} {\bibfnamefont {M.~S.}\ \bibnamefont
  {Kaznady}}\ and\ \bibinfo {author} {\bibfnamefont {D.~F.~V.}\ \bibnamefont
  {James}},\ }\href {\doibase 10.1103/PhysRevA.79.022109} {\bibfield  {journal}
  {\bibinfo  {journal} {Phys. Rev. A}\ }\textbf {\bibinfo {volume} {79}},\
  \bibinfo {pages} {022109} (\bibinfo {year} {2009})}\BibitemShut {NoStop}%
\bibitem [{\citenamefont {Smolin}\ \emph {et~al.}(2012)\citenamefont {Smolin},
  \citenamefont {Gambetta},\ and\ \citenamefont
  {Smith}}]{smolin_efficient_2012}%
  \BibitemOpen
  \bibfield  {author} {\bibinfo {author} {\bibfnamefont {J.~A.}\ \bibnamefont
  {Smolin}}, \bibinfo {author} {\bibfnamefont {J.~M.}\ \bibnamefont
  {Gambetta}}, \ and\ \bibinfo {author} {\bibfnamefont {G.}~\bibnamefont
  {Smith}},\ }\href {\doibase 10.1103/PhysRevLett.108.070502} {\bibfield
  {journal} {\bibinfo  {journal} {Phys. Rev. Lett.}\ }\textbf {\bibinfo
  {volume} {108}},\ \bibinfo {pages} {070502} (\bibinfo {year}
  {2012})}\BibitemShut {NoStop}%
\bibitem [{\citenamefont {Alquier}\ \emph {et~al.}(2013)\citenamefont
  {Alquier}, \citenamefont {Butucea}, \citenamefont {Hebiri}, \citenamefont
  {Meziani},\ and\ \citenamefont {Morimae}}]{AlquierButucea_2013}%
  \BibitemOpen
  \bibfield  {author} {\bibinfo {author} {\bibfnamefont {P.}~\bibnamefont
  {Alquier}}, \bibinfo {author} {\bibfnamefont {C.}~\bibnamefont {Butucea}},
  \bibinfo {author} {\bibfnamefont {M.}~\bibnamefont {Hebiri}}, \bibinfo
  {author} {\bibfnamefont {K.}~\bibnamefont {Meziani}}, \ and\ \bibinfo
  {author} {\bibfnamefont {T.}~\bibnamefont {Morimae}},\ }\href
  {https://link.aps.org/doi/10.1103/PhysRevA.88.032113} {\bibfield  {journal}
  {\bibinfo  {journal} {Phys. Rev. A}\ }\textbf {\bibinfo {volume} {88}},\
  \bibinfo {pages} {032113} (\bibinfo {year} {2013})}\BibitemShut {NoStop}%
\bibitem [{\citenamefont {Butucea}\ \emph {et~al.}(2015)\citenamefont
  {Butucea}, \citenamefont {Guta},\ and\ \citenamefont
  {Kypraios}}]{butucea_spectral_2015}%
  \BibitemOpen
  \bibfield  {author} {\bibinfo {author} {\bibfnamefont {C.}~\bibnamefont
  {Butucea}}, \bibinfo {author} {\bibfnamefont {M.}~\bibnamefont {Guta}}, \
  and\ \bibinfo {author} {\bibfnamefont {T.}~\bibnamefont {Kypraios}},\ }\href
  {http://stacks.iop.org/1367-2630/17/i=11/a=113050} {\bibfield  {journal}
  {\bibinfo  {journal} {New Journal of Physics}\ }\textbf {\bibinfo {volume}
  {17}},\ \bibinfo {pages} {113050} (\bibinfo {year} {2015})}\BibitemShut
  {NoStop}%
\bibitem [{\citenamefont {Koltchinskii}\ and\ \citenamefont
  {Xia}(2015)}]{Koltchinskii_optimal_2015}%
  \BibitemOpen
  \bibfield  {author} {\bibinfo {author} {\bibfnamefont {V.}~\bibnamefont
  {Koltchinskii}}\ and\ \bibinfo {author} {\bibfnamefont {D.}~\bibnamefont
  {Xia}},\ }\href
  {http://www.jmlr.org/papers/volume16/koltchinskii15a/koltchinskii15a.pdf}
  {\bibfield  {journal} {\bibinfo  {journal} {J. Mach. Learn. Res.}\ }\textbf
  {\bibinfo {volume} {16}},\ \bibinfo {pages} {1757} (\bibinfo {year}
  {2015})}\BibitemShut {NoStop}%
\bibitem [{\citenamefont {Kueng}\ \emph {et~al.}(2016)\citenamefont {Kueng},
  \citenamefont {Zhu},\ and\ \citenamefont {Gross}}]{kueng_clifford_2016}%
  \BibitemOpen
  \bibfield  {author} {\bibinfo {author} {\bibfnamefont {R.}~\bibnamefont
  {Kueng}}, \bibinfo {author} {\bibfnamefont {H.}~\bibnamefont {Zhu}}, \ and\
  \bibinfo {author} {\bibfnamefont {D.}~\bibnamefont {Gross}},\ }\href
  {https://arxiv.org/abs/1610.08070} {\bibfield  {journal} {\bibinfo  {journal}
  {arXiv preprint arXiv:1610.08070}\ } (\bibinfo {year} {2016})}\BibitemShut
  {NoStop}%
\bibitem [{\citenamefont {Tropp}(2012)}]{tropp_user-friendly_2012}%
  \BibitemOpen
  \bibfield  {author} {\bibinfo {author} {\bibfnamefont {J.~A.}\ \bibnamefont
  {Tropp}},\ }\href {\doibase 10.1007/s10208-011-9099-z} {\bibfield  {journal}
  {\bibinfo  {journal} {Found. Comput. Math.}\ }\textbf {\bibinfo {volume}
  {12}},\ \bibinfo {pages} {389} (\bibinfo {year} {2012})}\BibitemShut
  {NoStop}%
\bibitem [{\citenamefont {Renes}\ \emph {et~al.}(2004)\citenamefont {Renes},
  \citenamefont {Blume-Kohout}, \citenamefont {Scott},\ and\ \citenamefont
  {Caves}}]{renes_symmetric_2004}%
  \BibitemOpen
  \bibfield  {author} {\bibinfo {author} {\bibfnamefont {J.~M.}\ \bibnamefont
  {Renes}}, \bibinfo {author} {\bibfnamefont {R.}~\bibnamefont {Blume-Kohout}},
  \bibinfo {author} {\bibfnamefont {A.~J.}\ \bibnamefont {Scott}}, \ and\
  \bibinfo {author} {\bibfnamefont {C.~M.}\ \bibnamefont {Caves}},\ }\href
  {https://aip.scitation.org/doi/abs/10.1063/1.1737053} {\bibfield  {journal}
  {\bibinfo  {journal} {J. Math. Phys.}\ }\textbf {\bibinfo {volume} {45}},\
  \bibinfo {pages} {2171} (\bibinfo {year} {2004})}\BibitemShut {NoStop}%
\bibitem [{\citenamefont {Klappenecker}\ and\ \citenamefont
  {Rotteler}(2005)}]{klappenecker_mutually_2005}%
  \BibitemOpen
  \bibfield  {author} {\bibinfo {author} {\bibfnamefont {A.}~\bibnamefont
  {Klappenecker}}\ and\ \bibinfo {author} {\bibfnamefont {M.}~\bibnamefont
  {Rotteler}},\ }in\ \href {\doibase 10.1109/ISIT.2005.1523643} {\emph
  {\bibinfo {booktitle} {International Symposium on Information Theory, 2005.
  {ISIT} 2005. Proceedings}}}\ (\bibinfo {year} {2005})\ pp.\ \bibinfo {pages}
  {1740 --1744}\BibitemShut {NoStop}%
\bibitem [{\citenamefont {Dankert}\ \emph {et~al.}(2009)\citenamefont
  {Dankert}, \citenamefont {Cleve}, \citenamefont {Emerson},\ and\
  \citenamefont {Livine}}]{dankert_exact_2009}%
  \BibitemOpen
  \bibfield  {author} {\bibinfo {author} {\bibfnamefont {C.}~\bibnamefont
  {Dankert}}, \bibinfo {author} {\bibfnamefont {R.}~\bibnamefont {Cleve}},
  \bibinfo {author} {\bibfnamefont {J.}~\bibnamefont {Emerson}}, \ and\
  \bibinfo {author} {\bibfnamefont {E.}~\bibnamefont {Livine}},\ }\href
  {\doibase 10.1103/PhysRevA.80.012304} {\bibfield  {journal} {\bibinfo
  {journal} {Phys. Rev. A}\ }\textbf {\bibinfo {volume} {80}},\ \bibinfo
  {pages} {012304} (\bibinfo {year} {2009})}\BibitemShut {NoStop}%
\bibitem [{\citenamefont {Gross}\ \emph {et~al.}(2007)\citenamefont {Gross},
  \citenamefont {Audenaert},\ and\ \citenamefont {Eisert}}]{gross_evenly_2007}%
  \BibitemOpen
  \bibfield  {author} {\bibinfo {author} {\bibfnamefont {D.}~\bibnamefont
  {Gross}}, \bibinfo {author} {\bibfnamefont {K.}~\bibnamefont {Audenaert}}, \
  and\ \bibinfo {author} {\bibfnamefont {J.}~\bibnamefont {Eisert}},\ }\href
  {\doibase doi:10.1063/1.2716992} {\bibfield  {journal} {\bibinfo  {journal}
  {J. Math. Phys.}\ }\textbf {\bibinfo {volume} {48}},\ \bibinfo {pages}
  {052104} (\bibinfo {year} {2007})}\BibitemShut {NoStop}%
\bibitem [{\citenamefont {Halko}\ \emph {et~al.}(2011)\citenamefont {Halko},
  \citenamefont {Martinsson},\ and\ \citenamefont
  {Tropp}}]{halko_finding_2011}%
  \BibitemOpen
  \bibfield  {author} {\bibinfo {author} {\bibfnamefont {N.}~\bibnamefont
  {Halko}}, \bibinfo {author} {\bibfnamefont {P.}~\bibnamefont {Martinsson}}, \
  and\ \bibinfo {author} {\bibfnamefont {J.}~\bibnamefont {Tropp}},\ }\href
  {\doibase 10.1137/090771806} {\bibfield  {journal} {\bibinfo  {journal} {SIAM
  Review}\ }\textbf {\bibinfo {volume} {53}},\ \bibinfo {pages} {217} (\bibinfo
  {year} {2011})},\ \Eprint
  {http://arxiv.org/abs/https://doi.org/10.1137/090771806}
  {https://doi.org/10.1137/090771806} \BibitemShut {NoStop}%
\bibitem [{\citenamefont {Kalev}\ \emph {et~al.}(2015)\citenamefont {Kalev},
  \citenamefont {Kosut},\ and\ \citenamefont {Deutsch}}]{kalev_quantum_2015}%
  \BibitemOpen
  \bibfield  {author} {\bibinfo {author} {\bibfnamefont {A.}~\bibnamefont
  {Kalev}}, \bibinfo {author} {\bibfnamefont {R.~L.}\ \bibnamefont {Kosut}}, \
  and\ \bibinfo {author} {\bibfnamefont {I.~H.}\ \bibnamefont {Deutsch}},\
  }\href {https://www.nature.com/articles/npjqi201518} {\bibfield  {journal}
  {\bibinfo  {journal} {NPJ Quantum Inf.}\ }\textbf {\bibinfo {volume} {1}},\
  \bibinfo {pages} {15018} (\bibinfo {year} {2015})}\BibitemShut {NoStop}%
\bibitem [{\citenamefont {Kabanava}\ \emph {et~al.}(2016)\citenamefont
  {Kabanava}, \citenamefont {Kueng}, \citenamefont {Rauhut},\ and\
  \citenamefont {Terstiege}}]{kabanava_stable_2016}%
  \BibitemOpen
  \bibfield  {author} {\bibinfo {author} {\bibfnamefont {M.}~\bibnamefont
  {Kabanava}}, \bibinfo {author} {\bibfnamefont {R.}~\bibnamefont {Kueng}},
  \bibinfo {author} {\bibfnamefont {H.}~\bibnamefont {Rauhut}}, \ and\ \bibinfo
  {author} {\bibfnamefont {U.}~\bibnamefont {Terstiege}},\ }\href
  {http://dx.doi.org/10.1093/imaiai/iaw014} {\bibfield  {journal} {\bibinfo
  {journal} {Inf. Inference}\ }\textbf {\bibinfo {volume} {5}},\ \bibinfo
  {pages} {405} (\bibinfo {year} {2016})}\BibitemShut {NoStop}%
\bibitem [{\citenamefont {Scholten}\ and\ \citenamefont
  {Blume-Kohout}(2018)}]{scholten_behavior_2018}%
  \BibitemOpen
  \bibfield  {author} {\bibinfo {author} {\bibfnamefont {T.~L.}\ \bibnamefont
  {Scholten}}\ and\ \bibinfo {author} {\bibfnamefont {R.}~\bibnamefont
  {Blume-Kohout}},\ }\href {http://stacks.iop.org/1367-2630/20/i=2/a=023050}
  {\bibfield  {journal} {\bibinfo  {journal} {New J. Phys.}\ }\textbf {\bibinfo
  {volume} {20}},\ \bibinfo {pages} {023050} (\bibinfo {year}
  {2018})}\BibitemShut {NoStop}%
\bibitem [{\citenamefont {O'Donnell}\ and\ \citenamefont
  {Wright}(2016)}]{odonnel_efficient_2016}%
  \BibitemOpen
  \bibfield  {author} {\bibinfo {author} {\bibfnamefont {R.}~\bibnamefont
  {O'Donnell}}\ and\ \bibinfo {author} {\bibfnamefont {J.}~\bibnamefont
  {Wright}},\ }in\ \href {\doibase 10.1145/2897518.2897544} {\emph {\bibinfo
  {booktitle} {Proceedings of the Forty-eighth Annual ACM Symposium on Theory
  of Computing}}},\ \bibinfo {series and number} {STOC '16}\ (\bibinfo
  {publisher} {ACM},\ \bibinfo {address} {New York, NY, USA},\ \bibinfo {year}
  {2016})\ pp.\ \bibinfo {pages} {899--912}\BibitemShut {NoStop}%
\bibitem [{\citenamefont {Scott}(2006)}]{scott_tight_2006}%
  \BibitemOpen
  \bibfield  {author} {\bibinfo {author} {\bibfnamefont {A.~J.}\ \bibnamefont
  {Scott}},\ }\href {http://stacks.iop.org/0305-4470/39/i=43/a=009} {\bibfield
  {journal} {\bibinfo  {journal} {J. Phys. A}\ }\textbf {\bibinfo {volume}
  {39}},\ \bibinfo {pages} {13507} (\bibinfo {year} {2006})}\BibitemShut
  {NoStop}%
\bibitem [{\citenamefont {Gross}\ \emph {et~al.}(2015)\citenamefont {Gross},
  \citenamefont {Krahmer},\ and\ \citenamefont {Kueng}}]{gross_partial_2015}%
  \BibitemOpen
  \bibfield  {author} {\bibinfo {author} {\bibfnamefont {D.}~\bibnamefont
  {Gross}}, \bibinfo {author} {\bibfnamefont {F.}~\bibnamefont {Krahmer}}, \
  and\ \bibinfo {author} {\bibfnamefont {R.}~\bibnamefont {Kueng}},\ }\href
  {\doibase 10.1007/s00041-014-9361-2} {\bibfield  {journal} {\bibinfo
  {journal} {J. Fourier Anal. Appl.}\ }\textbf {\bibinfo {volume} {21}},\
  \bibinfo {pages} {229} (\bibinfo {year} {2015})}\BibitemShut {NoStop}%
\bibitem [{\citenamefont {Schwinger}(1960)}]{schwinger_unitary_1960}%
  \BibitemOpen
  \bibfield  {author} {\bibinfo {author} {\bibfnamefont {J.}~\bibnamefont
  {Schwinger}},\ }\href@noop {} {\bibfield  {journal} {\bibinfo  {journal}
  {Proc. Natl. Acad. Sci. USA}\ }\textbf {\bibinfo {volume} {46}},\ \bibinfo
  {pages} {570} (\bibinfo {year} {1960})}\BibitemShut {NoStop}%
\bibitem [{\citenamefont {Nielsen}\ and\ \citenamefont
  {Chuang}(2011)}]{nielsen_quantum_2011}%
  \BibitemOpen
  \bibfield  {author} {\bibinfo {author} {\bibfnamefont {M.~A.}\ \bibnamefont
  {Nielsen}}\ and\ \bibinfo {author} {\bibfnamefont {I.~L.}\ \bibnamefont
  {Chuang}},\ }\href@noop {} {\emph {\bibinfo {title} {Quantum Computation and
  Quantum Information: 10th Anniversary Edition}}},\ \bibinfo {edition} {10th}\
  ed.\ (\bibinfo  {publisher} {Cambridge University Press},\ \bibinfo {address}
  {New York, NY, USA},\ \bibinfo {year} {2011})\BibitemShut {NoStop}%
\bibitem [{\citenamefont {Kueng}\ and\ \citenamefont
  {Gross}(2015)}]{kueng_qubit_2015}%
  \BibitemOpen
  \bibfield  {author} {\bibinfo {author} {\bibfnamefont {R.}~\bibnamefont
  {Kueng}}\ and\ \bibinfo {author} {\bibfnamefont {D.}~\bibnamefont {Gross}},\
  }\href {https://arxiv.org/abs/1510.02767} {\bibfield  {journal} {\bibinfo
  {journal} {preprint arXiv:1510.02767}\ } (\bibinfo {year}
  {2015})}\BibitemShut {NoStop}%
\bibitem [{\citenamefont {Zhu}(2017)}]{zhu_clifford_2017}%
  \BibitemOpen
  \bibfield  {author} {\bibinfo {author} {\bibfnamefont {H.}~\bibnamefont
  {Zhu}},\ }\href {\doibase 10.1103/PhysRevA.96.062336} {\bibfield  {journal}
  {\bibinfo  {journal} {Phys. Rev. A}\ }\textbf {\bibinfo {volume} {96}},\
  \bibinfo {pages} {062336} (\bibinfo {year} {2017})}\BibitemShut {NoStop}%
\bibitem [{\citenamefont {Webb}(2015)}]{webb_clifford_2015}%
  \BibitemOpen
  \bibfield  {author} {\bibinfo {author} {\bibfnamefont {Z.}~\bibnamefont
  {Webb}},\ }\href {https://arxiv.org/abs/1510.02769} {\bibfield  {journal}
  {\bibinfo  {journal} {arXiv preprint arXiv:1510.02769}\ } (\bibinfo {year}
  {2015})}\BibitemShut {NoStop}%
\bibitem [{\citenamefont {Tomczak-Jaegermann}(1974)}]{tomczak_moduli_1974}%
  \BibitemOpen
  \bibfield  {author} {\bibinfo {author} {\bibfnamefont {N.}~\bibnamefont
  {Tomczak-Jaegermann}},\ }\href {http://eudml.org/doc/217886} {\bibfield
  {journal} {\bibinfo  {journal} {Stud. Math.}\ }\textbf {\bibinfo {volume}
  {50}},\ \bibinfo {pages} {163} (\bibinfo {year} {1974})}\BibitemShut
  {NoStop}%
\bibitem [{\citenamefont {Lust-Piquard}(1986)}]{lust_inegalites1986}%
  \BibitemOpen
  \bibfield  {author} {\bibinfo {author} {\bibfnamefont {F.}~\bibnamefont
  {Lust-Piquard}},\ }\href@noop {} {\bibfield  {journal} {\bibinfo  {journal}
  {CR Acad. Sci. Paris}\ }\textbf {\bibinfo {volume} {303}},\ \bibinfo {pages}
  {289} (\bibinfo {year} {1986})}\BibitemShut {NoStop}%
\bibitem [{\citenamefont {Pisier}\ and\ \citenamefont
  {Xu}(1997)}]{pisier_noncommutative_1997}%
  \BibitemOpen
  \bibfield  {author} {\bibinfo {author} {\bibfnamefont {G.}~\bibnamefont
  {Pisier}}\ and\ \bibinfo {author} {\bibfnamefont {Q.}~\bibnamefont {Xu}},\
  }\href {\doibase 10.1007/s002200050224} {\bibfield  {journal} {\bibinfo
  {journal} {Commun. Math. Phys.}\ }\textbf {\bibinfo {volume} {189}},\
  \bibinfo {pages} {667} (\bibinfo {year} {1997})}\BibitemShut {NoStop}%
\bibitem [{\citenamefont {Rudelson}(1999)}]{rudelson_random_1999}%
  \BibitemOpen
  \bibfield  {author} {\bibinfo {author} {\bibfnamefont {M.}~\bibnamefont
  {Rudelson}},\ }\href {\doibase https://doi.org/10.1006/jfan.1998.3384}
  {\bibfield  {journal} {\bibinfo  {journal} {J. Funct. Anal.}\ }\textbf
  {\bibinfo {volume} {164}},\ \bibinfo {pages} {60 } (\bibinfo {year}
  {1999})}\BibitemShut {NoStop}%
\bibitem [{\citenamefont {Ahlswede}\ and\ \citenamefont
  {Winter}(2002)}]{ahlswede_strong_2002}%
  \BibitemOpen
  \bibfield  {author} {\bibinfo {author} {\bibfnamefont {R.}~\bibnamefont
  {Ahlswede}}\ and\ \bibinfo {author} {\bibfnamefont {A.}~\bibnamefont
  {Winter}},\ }\href {\doibase 10.1109/18.985947} {\bibfield  {journal}
  {\bibinfo  {journal} {IEEE Trans. Inform. Theory}\ }\textbf {\bibinfo
  {volume} {48}},\ \bibinfo {pages} {569} (\bibinfo {year} {2002})}\BibitemShut
  {NoStop}%
\bibitem [{\citenamefont {Gross}(2011)}]{gross_recovering_2011}%
  \BibitemOpen
  \bibfield  {author} {\bibinfo {author} {\bibfnamefont {D.}~\bibnamefont
  {Gross}},\ }\href {\doibase 10.1109/TIT.2011.2104999} {\bibfield  {journal}
  {\bibinfo  {journal} {IEEE Trans. Inform. Theory}\ }\textbf {\bibinfo
  {volume} {57}},\ \bibinfo {pages} {1548} (\bibinfo {year}
  {2011})}\BibitemShut {NoStop}%
\bibitem [{\citenamefont {Oliveira}\ \emph {et~al.}(2010)\citenamefont
  {Oliveira} \emph {et~al.}}]{oliveira_sums_2010}%
  \BibitemOpen
  \bibfield  {author} {\bibinfo {author} {\bibfnamefont {R.~I.}\ \bibnamefont
  {Oliveira}} \emph {et~al.},\ }\href
  {http://emis.ams.org/journals/EJP-ECP/article/download/1544/1889.pdf}
  {\bibfield  {journal} {\bibinfo  {journal} {Electron. Commun. Probab}\
  }\textbf {\bibinfo {volume} {15}},\ \bibinfo {pages} {26} (\bibinfo {year}
  {2010})}\BibitemShut {NoStop}%
\bibitem [{\citenamefont {Tropp}(2018)}]{tropp_second-order_2018}%
  \BibitemOpen
  \bibfield  {author} {\bibinfo {author} {\bibfnamefont {J.~A.}\ \bibnamefont
  {Tropp}},\ }\href {\doibase https://doi.org/10.1016/j.acha.2016.07.005}
  {\bibfield  {journal} {\bibinfo  {journal} {Applied and Computational
  Harmonic Analysis}\ }\textbf {\bibinfo {volume} {44}},\ \bibinfo {pages} {700
  } (\bibinfo {year} {2018})}\BibitemShut {NoStop}%
\bibitem [{\citenamefont {Foucart}\ and\ \citenamefont
  {Rauhut}(2013)}]{foucart_mathematical_2013}%
  \BibitemOpen
  \bibfield  {author} {\bibinfo {author} {\bibfnamefont {S.}~\bibnamefont
  {Foucart}}\ and\ \bibinfo {author} {\bibfnamefont {H.}~\bibnamefont
  {Rauhut}},\ }\href@noop {} {\emph {\bibinfo {title} {A Mathematical
  Introduction to Compressive Sensing}}}\ (\bibinfo  {publisher} {Birkh\"auser
  Basel},\ \bibinfo {year} {2013})\BibitemShut {NoStop}%
\bibitem [{\citenamefont {Lancien}\ and\ \citenamefont
  {Winter}(2017)}]{lancien_approximating_2017}%
  \BibitemOpen
  \bibfield  {author} {\bibinfo {author} {\bibfnamefont {C.}~\bibnamefont
  {Lancien}}\ and\ \bibinfo {author} {\bibfnamefont {A.}~\bibnamefont
  {Winter}},\ }\href {https://arxiv.org/abs/1711.00697} {\bibfield  {journal}
  {\bibinfo  {journal} {preprint arXiv:1711.00697}\ } (\bibinfo {year}
  {2017})}\BibitemShut {NoStop}%
\bibitem [{\citenamefont {Talagrand}(2006)}]{talagrand_generic_2006}%
  \BibitemOpen
  \bibfield  {author} {\bibinfo {author} {\bibfnamefont {M.}~\bibnamefont
  {Talagrand}},\ }\href@noop {} {\emph {\bibinfo {title} {The generic chaining:
  upper and lower bounds of stochastic processes}}}\ (\bibinfo  {publisher}
  {Springer Science \& Business Media},\ \bibinfo {year} {2006})\BibitemShut
  {NoStop}%
\bibitem [{\citenamefont {Vershynin}(2012)}]{vershynin_introduction_2012}%
  \BibitemOpen
  \bibfield  {author} {\bibinfo {author} {\bibfnamefont {R.}~\bibnamefont
  {Vershynin}},\ }in\ \href {\doibase 10.1017/CBO9780511794308.006} {\emph
  {\bibinfo {booktitle} {Compressed Sensing: Theory and Practice}}},\ \bibinfo
  {editor} {edited by\ \bibinfo {editor} {\bibfnamefont {Y.~C.}\ \bibnamefont
  {Eldar}}\ and\ \bibinfo {editor} {\bibfnamefont {G.}~\bibnamefont
  {Kutyniok}}}\ (\bibinfo  {publisher} {Cambridge University Press},\ \bibinfo
  {year} {2012})\ pp.\ \bibinfo {pages} {210--268}\BibitemShut {NoStop}%
\bibitem [{\citenamefont {Wootters}\ and\ \citenamefont
  {Fields}(1989)}]{wootters_optimal_1989}%
  \BibitemOpen
  \bibfield  {author} {\bibinfo {author} {\bibfnamefont {W.~K.}\ \bibnamefont
  {Wootters}}\ and\ \bibinfo {author} {\bibfnamefont {B.~D.}\ \bibnamefont
  {Fields}},\ }\href {\doibase https://doi.org/10.1016/0003-4916(89)90322-9}
  {\bibfield  {journal} {\bibinfo  {journal} {Ann. Phys.}\ }\textbf {\bibinfo
  {volume} {191}},\ \bibinfo {pages} {363 } (\bibinfo {year}
  {1989})}\BibitemShut {NoStop}%
\bibitem [{\citenamefont {Bandyopadhyay}\ \emph {et~al.}(2002)\citenamefont
  {Bandyopadhyay}, \citenamefont {Boykin}, \citenamefont {Roychowdhury},\ and\
  \citenamefont {Vatan}}]{bandyopadhyay_new_2002}%
  \BibitemOpen
  \bibfield  {author} {\bibinfo {author} {\bibnamefont {Bandyopadhyay}},
  \bibinfo {author} {\bibnamefont {Boykin}}, \bibinfo {author} {\bibnamefont
  {Roychowdhury}}, \ and\ \bibinfo {author} {\bibnamefont {Vatan}},\ }\href
  {\doibase 10.1007/s00453-002-0980-7} {\bibfield  {journal} {\bibinfo
  {journal} {Algorithmica}\ }\textbf {\bibinfo {volume} {34}},\ \bibinfo
  {pages} {512} (\bibinfo {year} {2002})}\BibitemShut {NoStop}%
\bibitem [{\citenamefont {Klappenecker}\ and\ \citenamefont
  {R{\"o}tteler}(2004)}]{klappenecker_constructions_2004}%
  \BibitemOpen
  \bibfield  {author} {\bibinfo {author} {\bibfnamefont {A.}~\bibnamefont
  {Klappenecker}}\ and\ \bibinfo {author} {\bibfnamefont {M.}~\bibnamefont
  {R{\"o}tteler}},\ }in\ \href@noop {} {\emph {\bibinfo {booktitle} {Finite
  Fields and Applications}}},\ \bibinfo {editor} {edited by\ \bibinfo {editor}
  {\bibfnamefont {G.~L.}\ \bibnamefont {Mullen}}, \bibinfo {editor}
  {\bibfnamefont {A.}~\bibnamefont {Poli}}, \ and\ \bibinfo {editor}
  {\bibfnamefont {H.}~\bibnamefont {Stichtenoth}}}\ (\bibinfo  {publisher}
  {Springer},\ \bibinfo {address} {Berlin, Heidelberg},\ \bibinfo {year}
  {2004})\ pp.\ \bibinfo {pages} {137--144}\BibitemShut {NoStop}%
\end{thebibliography}%

\cleardoublepage

\part*{Appendix}

At the heart of this work is \emph{projected least squares} (PLS) -- a simple point estimator for quantum state tomography from tomographically complete measurements $\left\{M_1,\ldots,M_m \right\} \subset \mathbb{H}_d$. PLS is a three-step procedure, see also Tab.~\ref{tab:box}:
\begin{enumerate}
\item Estimate outcome probabilities 
by frequencies. 
\item Construct the least squares (linear inversion) estimator:
\begin{align}
\hat{L}_n = \underset{X \in \mathbb{H}_d}{\textrm{argmin}} \quad \sum_{i=1}^m \left( f_i - \mathrm{tr}(M_i X ) \right)^2 .
\label{eq:linear_inversion_appendix}
\end{align}
\item Project onto the set of all quantum states:
\begin{equation}
\hat{\rho}_n = \underset{\textrm{$\sigma$ is a quantum state}}{\textrm{argmin}} \| \hat{L}_n - \sigma \|_2. \label{eq:pli_appendix}
\end{equation}
\end{enumerate}
We analyze the performance of PLS for a variety of concrete measurement scenarios: \emph{structured POVMs}, \emph{Pauli observables}, \emph{Pauli basis measurements} and the \emph{uniform POVM}. For each of them, $\hat{\rho}_n$ may be equipped with rigorous non-asymptotic confidence regions in trace distance. 
In this appendix, we complement the rather succinct presentation in the main text with additional explanations, motivations and more detailed arguments. 

\paragraph*{\textbf{Outline:}} 

In Section~\ref{sec:LI_estimators} we provide explicit least squares solutions \eqref{eq:linear_inversion_appendix} for the different measurements. We also review essential features and properties of the individual scenarios to provide context.

Section~\ref{sec:concentration} contains the main conceptual insight of this work: least squares estimators may be interpreted as sums of independent random matrices --the randomness is due to the fundamental laws of quantum mechanics (Born's rule). This allows us to apply strong matrix-valued concentration inequalities to show that, with hight probability,  $\hat{L}_n$ is close to the true target state in operator norm.

Section~\ref{sec:conversion} is devoted to showing that closeness of $\hat{L}_n$ in operator norm implies closeness of $\hat{\rho}_n$ in trace norm. 

We combine these two insights in Section~\ref{sec:effective_rank} to arrive at the main result of this work: convergence guarantees for the PLS estimator in trace norm. The result derived there is a strict generalization of Theorem~\ref{thm:main_result} quoted in the main text. It extends to the notion of \emph{effective rank} which may be beneficial in concrete applications. We illustrate this potential benefit with a caricature of a faulty state preparation apparatus.

In Section~\ref{sec:uniform} yet stronger convergence guarantees for the uniform POVM are derived. The proof technique is completely different and we believe that it may be of independent interest to the community.

Finally, we present additional numerical experiments in Section~\ref{sec:numerics}.

\section{Closed-form expressions for least squares estimators} \label{sec:LI_estimators}

As outlined in the main text, any POVM measurement can be viewed as a linear map $\mathcal{M}: \mathbb{H}_d \to \mathbb{R}^m$, defined component-wise as $\left[ \mathcal{M}(X) \right]_i = \mathrm{tr}(M_i X)$ for $i \in \left[m \right]$. This map is injective if and only if the measurement is tomographically complete. Provided that this is the case, the least squares estimator \eqref{eq:linear_inversion_appendix} admits a unique solution:
\begin{equation*}
\hat{L}_n = \left( \mathcal{M}^\dagger \mathcal{M} \right)^{-1} \mathcal{M}^\dagger (f_n),
\end{equation*}
where $f_n \in \mathbb{R}^m$ subsumes the individual frequency estimates. In this section, we evaluate this formula explicitly for different types of prominent measurements.

\subsection{The uniform POVM and $2$-designs}

The \emph{uniform/covariant} POVM in $d$-dimensions corresponds to the union of all  (properly re-normalized) rank-one projectors: $\left\{ d | v \rangle \! \langle v| \mathrm{d} v \right\}_{v \in \mathbb{S}^d}$. Here, $\mathrm{d} v$ denotes the unique, unitarily invariant, measure on the complex unit sphere induced by the Haar measure (over the unitary group $U(d)$). 
Its high degree of symmetry allows for analyzing this POVM by means of powerful tools from representation theory. 
This is widely known, see e.g.\ \cite{scott_tight_2006,gross_partial_2015}, but we include a short presentation here to be self-contained.
Define the \emph{frame operator of order $k$}: $F_{(k)} = \int_{\mathbb{S}^d} \left( |v \rangle \! \langle v| \right)^{\otimes k} \mathrm{d}v \in  \mathbb{H}_d^{\otimes k}$. Unitary invariance of $\mathrm{d} v$ implies that this frame operator commutes with every $k$-fold tensor product of a unitary matrix $U \in U(d)$:
\begin{align*}
U^{\otimes k} F_{(k)}
=& \int_{\mathbb{S}^d} \left( U|v \rangle \! \langle v| \right)^{\otimes k} \mathrm{d} v = \int_{\mathbb{S}^d} \left(|\tilde{v} \rangle \! \langle \tilde{v} | U \right)^{\otimes k} \mathrm{d} \tilde{v} \\
=& F_{(k)} U^{\otimes k}.
\end{align*}
Here, we have used a change of variables ($\tilde{v} = U v$) together with the fact that $\mathrm{d}v$ is unitarily invariant ($\mathrm{d} \tilde{v} = \mathrm{d} v$). 
Schur's Lemma -- one of the most fundamental tools in representation theory -- states that any matrix that commutes with every element of a given group representation must be proportional to a sum of the projectors onto the associated irreducible representations (irreps). 
For the task at hand, the representation of interest is the diagonal representation of the unitary group: $U \mapsto U^{\otimes k}$ for all $U \in U(d)$. This representation affords, in general, many irreps that may be characterized using Schur-Weyl duality. 
The \emph{symmetric subspace} $\mathrm{Sym}^{(k)} \subset \left( \mathbb{C}^d \right)^{\otimes k}$ is one of them and corresponds to the subspace of all vectors that are invariant under permuting tensor factors. 
Crucially, $F_{(k)}$ is an average over rank-one projectors onto vectors $|v \rangle^{\otimes k} \in \mathrm{Sym}^{(k)}$ and, therefore, its range must be contained entirely within $\mathrm{Sym}^{(k)}$. Combining this with the assertion of Schur's lemma then yields
\begin{equation}
F_{(k)} = \int_{\mathbb{S}^d} \left( |v \rangle \! \langle v| \right)^{\otimes k} \mathrm{d} v = \binom{d+k-1}{k}^{-1} P_{\mathrm{Sym}^{(k)}} \quad k \in \mathbb{N}, \label{eq:frame_operator}
\end{equation}
The pre-factor $\binom{d+k-1}{k}^{-1} = \mathrm{dim} \left( \mathrm{Sym}^{(k)} \right)^{-1}$ follows from the fact that $F_{(k)}$ has unit trace. 

This closed-form expression is very useful. In particular, it implies that the uniform POVM $\left\{ d| v \rangle \! \langle v| \right\}_{v \in \mathbb{S}^d}$ is almost an isometry. Fix $X \in \mathbb{H}_d$ and compute
\begin{align}
& (d+1)\int_{\mathbb{S}^d} d \langle v| X |v \rangle |v \rangle \! \langle v| \mathrm{d} v \nonumber \\
=& (d+1) d \mathrm{tr}_1 \left( X \otimes \mathbb{I} \int_{\mathbb{S}^d} \left( |v \rangle \! \langle v| \right)^{\otimes 2} \mathrm{d}v \right) \nonumber \\
=& 2 \mathrm{tr}_1 \left( X \otimes \mathbb{I} P_{\mathrm{Sym}^{(2)}} \right),
\label{eq:uniform_aux1}
\end{align}
where $\mathrm{tr}_1 (A \otimes B) = \mathrm{tr}(A) B$ denotes the partial trace over the first tensor factor. The projector onto the totally symmetric subspace of two parties has an explicit representation: $P_{\mathrm{Sym}^{(2)}} = \frac{1}{2} \left( \mathbb{I} + \mathbb{F} \right)$, where $\mathbb{F}$
 denotes the \emph{flip operator}, i.e.\ $\mathbb{F} |x \rangle \otimes |y \rangle = |y \rangle \otimes |x \rangle$ for all $|x \rangle, |y \rangle \in \mathbb{C}^d$ and extend it linearly to the entire tensor product. Inserting this explicit characterization into Eq.~\eqref{eq:uniform_aux1} yields
\begin{align}
 (d+1)\int_{\mathbb{S}^d} d \langle v| X |v \rangle |v \rangle \! \langle v| \mathrm{d} v 
=& \mathrm{tr}_1 \left( X \otimes \mathbb{I} \left(\mathbb{I} + \mathbb{F} \right) \right) \nonumber\\
=& X + \mathrm{tr}(X) \mathbb{I}. \label{eq:2design_appendix}
\end{align}
We emphasize that the full symmetry of the uniform POVM is not required to derive this formula: Eq.~\eqref{eq:frame_operator} for $k=2$ is sufficient. 
This motivates the following definition:

\begin{definition}[2-design]
A (finite) set of $m$ rank-one projectors $\left\{ |v_i \rangle \! \langle v_i | \right\}_{i=1}^m$ is called a \emph{(complex-projective) 2-design} if
\begin{align*}
\frac{1}{m} \sum_{i=1}^m \left( |v_i \rangle \! \langle v_i| \right)^{\otimes 2} = \binom{d+1}{2}^{-1} P_{\mathrm{Sym}^{(2)}}.
\end{align*}
\end{definition}
Taking the partial trace of this expression yields
\begin{equation*}
\frac{1}{m} \sum_{i=1}^m |v_i \rangle \! \langle v_i| = \frac{1}{d} \mathbb{I},
\end{equation*}
highlighting that each 2-design is proportional to a POVM $\mathcal{M}=  \left\{ \frac{d}{m} |v_i \rangle \! \langle v_i| \right\}_{i=1}^m$. Moreover, viewed as a map $\mathcal{M}: \mathbb{H}_d \to \mathbb{R}^m$, every such POVM obeys
\begin{align*}
\mathcal{M}^\dagger \mathcal{M} (X)
= \frac{d^2}{m^2} \sum_{i=1}^m \langle v_i| X |v_i \rangle |v_i \rangle \! \langle v_i| 
= \frac{d\left( X+ \mathrm{tr}(X) \mathbb{I} \right)}{(d+1)m}
\end{align*}
for any $X \in \mathbb{H}_d$, which can be readily inverted:
\begin{align}
\left( \mathcal{M}^\dagger \mathcal{M} \right)^{-1} (X)
= \frac{m}{d} \left((d+1)X - \mathrm{tr}(X) \mathbb{I} \right).
\label{eq:near_isometry} 
\end{align}
Inserting this formula into the closed-form expression of the linear-inversion estimator yields
\begin{align}
\hat{L}_n =& \left( \mathcal{M}^\dagger \mathcal{M} \right)^{-1}
\left( \mathcal{M}^\dagger (f) \right) \\
=&\left(\mathcal{M}^\dagger \mathcal{M}\right)^{-1}\left( \frac{d}{m} \sum_{i=1}^m f_i |v_i \rangle \! \langle v_i| \right) \nonumber \\
=& \sum_{i=1}^m f_i \left( (d+1) |v_i \rangle \! \langle v_i| - \mathrm{tr} \left( |v_i \rangle \! \langle v_i| \right) \mathbb{I} \right) \nonumber \\
=& (d+1) \sum_{i=1}^m f_i |v_i \rangle \! \langle v_i| - \mathbb{I}
\label{eq:LI_2design_appendix}
\end{align}
for any frequency vector $f_n \in \mathbb{R}^m$.
Mathematically, this is a consequence of the fact that 2-design POVMs ``almost'' form a tight frame on $\mathbb{H}_d$. 
The close connection to well-behaved, tomographically complete, rank-one POVMs has spurred considerable interest in the identification of 2-designs. 
Over the past decades, the following concrete examples have been identified:
\begin{itemize}
\item[(i)] \emph{Equiangular lines (SIC POVMs):} a family of $m$ unit vectors $|v_1 \rangle, \ldots,| v_m \rangle \in \mathbb{S}^d$ is equiangular, if $\left| \langle v_i, v_j \rangle \right|^2$ is constant for all $i \neq j$.  
The maximal cardinality of such a set is $m=d^2$ in which case the angle must be fixed: $\left| \langle v_i, v_j \rangle \right|^2=\frac{1}{d+1}$. Such maximal sets of equiangular lines are known to form 2-designs \cite{renes_symmetric_2004} and have been termed \emph{symmetric, informationally complete (SIC) POVMs}. This nomenclature underlines the importance of Eq.~\eqref{eq:2design_appendix} for the original quantum motivation of the study of equiangular lines. While several explicit constructions of SIC POVMs exist, the general question of their existence remains an intriguing open problem.

\item[(ii)] \emph{Mutually unbiased bases (MUBs):} 
Two orthonormal bases $\left\{ |b_i \rangle \right\}_{i=1}^d$ and $\left\{ |c_i \rangle \right\}_{i=1}^d$ of $\mathbb{C}^d$ are \emph{mutually unbiased} if $\left| \langle b_i, c_j \rangle \right|^2 = \frac{1}{d}$ for all $1 \leq i,j \leq d$. 
The study of such mutually unbiased bases (MUBs) has a rich history in quantum mechanics that dates back to Schwinger \cite{schwinger_unitary_1960}. It is known that at most $(d+1)$ pairwise mutually unbiased bases can exist in dimension $d$ and explicit algebraic constructions are known for prime power dimensions ($d=p^k$). 
Klappenecker and Roettler \cite{klappenecker_mutually_2005} showed that maximal sets of MUBs are guaranteed to form 2-designs. 

\item[(iii)] \emph{stabilizer states (STABs):} the stabilizer formalism is one of the cornerstones of quantum computation, fault tolerance and error correction, see e.g.\ \cite{nielsen_quantum_2011}. 
Let $\mathcal{P}_k$ be the Pauli group on $k$ qubits ($d=2^k$), i.e.\ the group generated by $k$-fold tensor products of the elementary Pauli matrices. 
It is then possible to find maximal abelian subgroups $\mathcal{S} \subset \mathcal{P}_k$ of size $d=2^k$. Since all matrices $W \in \mathcal{S}$ commute, they can be simultaneously diagonalized and determine a single unit vector which is the joint eigenvector with eigenvalue $+1$ of all the matrices in $\mathcal{S}$ (provided that $- \mathbb{I} \notin \mathcal{S}$). Such vectors are called \emph{stabilizer states} (STAB) and the group $\mathcal{S} \subset \mathcal{P}_k$ is its associated \emph{stabilizer group}. 
A total of $m = 2^k \prod_{i=0}^k \left( d^i+1 \right) = 2^{\frac{1}{2} k^2 + o(k)}$ different stabilizer states can be generated this way. The union of all of them is actually known to form a 3-design \cite{kueng_qubit_2015, zhu_clifford_2017, webb_clifford_2015} and, therefore, also a 2-design. The latter is also a consequence of earlier results \cite{gross_evenly_2007,dankert_exact_2009} 
\end{itemize}

\subsection{Pauli observables}

For $d=2^k$, the Pauli matrices $W_1,\ldots,W_{d^2} \in \mathbb{H}_d$ arise from all possible $k$-fold tensor products of elementary  Pauli matrices $\left\{\mathbb{I},\sigma_x,\sigma_y,\sigma_z \right\} \subset \mathbb{H}_2$. They are well-known to form a unitary operator basis:
\begin{equation}
X = \frac{1}{d} \sum_{i=1}^{d^2} \mathrm{tr} \left( W_i X \right) W_i, \label{eq:unitary_basis}
\end{equation}
for all $X \in \mathbb{H}_d$. 
While they do constitute observables, Pauli matrices by themselves are not POVMs. However, every observable $W_i$ may be associated with a two-outcome POVM $\mathcal{M}_i = \left\{P_i^\pm \right\} = \left\{\frac{1}{2} \left( \mathbb{I} + W_i \right) \right\}$. 
The union $\bigcup_{i=1}^{d^2} \mathcal{M}_i$ of all these 2-outcome POVMs consitutes a linear map $\mathcal{M}: \mathbb{H}_d \to \mathbb{R}^{2m}$ that obeys
\begin{align*}
\mathcal{M}^\dagger \mathcal{M}(X)
=& \sum_{i=1}^{d^2} \left( \mathrm{tr} \left( P^+_i X \right) P^+_i + \mathrm{tr} \left( P^-_i X \right) P^-_i \right) \\
=& \sum_{i=1}^{d^2} \frac{1}{2} \left(\mathrm{tr}(X) \mathbb{I} + \mathrm{tr}(W_i X ) W_i \right) \\
=& \frac{d}{2} \left( d\mathrm{tr}(X) \mathbb{I} +  X \right),
\end{align*}
where the last line is due to Eq.~\eqref{eq:unitary_basis}.
Once more, this expression can be readily inverted:
\begin{align*}
\left( \mathcal{M}^\dagger \mathcal{M} \right)^{-1} (X)
= \frac{2}{d} X - \frac{2\mathrm{tr}(X)}{d^2+1} \mathbb{I}.
\end{align*}
Before we continue, we note that one Pauli matrix is equal to the identity, say $W_1 = \mathbb{I}$, and the associated POVM is trivial. Hence, we suppose  that $n$ copies of $\rho$ are distributed equally among all $d^2-1$ non-trivial 2-Outcome POVMs $\mathcal{M}_i$. We denote the resulting frequencies by $\left[f\right]_i^\pm$ and suppress the dependence on the number of samples. Then, the explicit solution to the least squares problem becomes
\begin{align*}
\hat{L}_n =& \left( \mathcal{M}^\dagger \mathcal{M} \right)^{-1} \left( \mathcal{M}^\dagger (f_n) \right)  \\
=& \left( \mathcal{M}^\dagger \mathcal{M} \right)^{-1} (\mathbb{I}) + \sum_{i=2}^{d^2} \sum_{o=\pm}\left[f \right]_i^o \left( \mathcal{M}^\dagger \mathcal{M} \right)^{-1} \left( P_i^o \right)   \\
=& \frac{2}{d(d^2+1)}\mathbb{I}
+ \sum_{i=2}^{d^2} \left[ f \right]_i^+ \left( \frac{1}{d} \left( \mathbb{I} + W_i \right) - \frac{d}{d^2+1} \mathbb{I} \right) \\
+& \sum_{i=2}^{d^2} \left[ f \right]_i^- \left( \frac{1}{d} \left( \mathbb{I} - W_i \right) - \frac{d}{d^2+1} \right) \\
=& \frac{1}{d} \sum_{i=2}^{d^2} \left( \left[ f \right]_i^+ - \left[ f \right]_i^- \right) W_i 
+ \frac{2 + \sum_{i=2}^{d^2} \left( \left[ f \right]_i^+ + \left[ f \right]_i^-  \right)}{d(d^2+1)} \mathbb{I}.
\end{align*}
We can simplify this expression further by noticing that each 2-outcome POVM is dichotomic: either $+$ or $-$ is observed for every run. This implies $\left[ f \right]_i^+ + \left[ f \right]_i^- =1$ and, by extension, $\sum_{i=2}^{d^2} \left( \left[ f_i \right]_i^+ + \left[f_i \right]_i^- \right) = d^2-1$. Hence,
\begin{align}
\hat{L}_n =& \frac{1}{d} \sum_{i=2}^{d^2} \left( \left[ f \right]_i^+ - \left[ f \right]_i^- \right) W_i + \frac{1}{d} \mathbb{I}\nonumber \\
=&
 \frac{1}{d} \sum_{i=1}^{d^2} \left( \left[ f \right]_i^+ - \left[ f \right]_i^- \right)W_i, \label{eq:LI_pauli_appendix}
\end{align}
because $\left[f \right]_1^+ =1$. 
This is the formula from the main text and has a compelling interpretation: the difference $\hat{\mu}_i = \left[ f_i \right]_i^+ - \left[ f_i \right]_i^-$ is an empirical estimate for the expectation value $\mu_i = \mathrm{tr} \left( W_i X \right)$ of the $i$-th Pauli observable. Finally, note that this estimator is again unbiased with respect to random fluctuations in the sample statistics:
\begin{equation}
\mathbb{E} \left[ \hat{L}_n \right] = \frac{1}{d} \sum_{i=1}^d \mathrm{tr}\left(W_i \rho \right) W_i = \rho. \label{eq:pauli_unbiased}
\end{equation}

\subsection{Pauli basis measurements}

Before considering the general case, we find it instructive to consider the single qubit case in more detail. For now, fix $d=2$ and note that there are three non-trivial Pauli matrices $\sigma_s$ with $s \in \left\{ x,y,z \right\}$. We may associate each $\sigma_s$ with a 2-outcome POVM that is also a basis measurement: $\frac{1}{2}\left( \mathbb{I} \pm \sigma_s \right)=|b_\pm^{(s)} \rangle \! \langle b_\pm^{(s)}|$.
For $s \neq s'$
\begin{align*}
\left| \langle b_{\pm}^{(s)} , b_{o'}^{(\pm)} \rangle \right|^2 
= \frac{1}{4} \mathrm{tr} \left( \left( \mathbb{I}  \pm \sigma_s \right) \left( \mathbb{I}  \pm \sigma_{s'}\right) \right) 
= \frac{1}{2},
\end{align*}
because $\sigma_s,\sigma_{s'}$ and $\sigma_{s} \sigma_{s'}$ have vanishing trace. This implies that the six vectors $|b^{(s)}_{o} \rangle \! \langle b^{(s)}_{o}|$ with $o \in \left\{ \pm 1 \right\}$ form a maximal set of $3=(d+1)$ mutually unbiased bases. Such vector sets form spherical 2-designs and Eq.~\eqref{eq:2design_appendix} ensures for any $X \in \mathbb{H}_d$
\begin{equation}
 \sum_{s,o} \langle b^{(s)}_o | X |b^{(s)}_o \rangle |b^{(s)}_o \rangle \! \langle b^{(s)}_o| =  \left( X+ \mathrm{tr}(X) \mathbb{I}\right) = 3\mathcal{D}_{1/3}(X). \label{eq:single_qubit}
\end{equation}
Here, $\mathcal{D}_{1/3}: \mathbb{H}_d \to \mathbb{H}_d$ denotes a single-qubit depolarizing channel with loss parameter $p =\frac{1}{3}$. 

This behavior extends to multi-qubit systems, i.e.\ $d=2^k$. 
Suppose that we perform $k$ local (single-qubit) Pauli measurements on a $k$-qubit state $\rho \in \mathbb{H}_d$. 
Then, there are a total of $k$ potential combinations that we label by a string $\vec{s} =(s_1,\ldots,s_k) \in \left\{x,y,z \right\}^k$. 
Each of them corresponds to a POVM $\mathcal{M}^{(\vec{s})}$ with $2^k=d$ outcomes that we label by $\vec{o} =(o_1,\ldots,o_k) \in \left\{ \pm 1 \right\}^k$. 
The POVM element associated with index $\vec{s}$ and outcome $\vec{o}$ has an appealing tensor-product structure: $|b^{(\vec{s})}_{\vec{o}} \rangle \! \langle b^{(\vec{s})}_{\vec{o}}| = \bigotimes_{i=1}^k |b_{o_i}^{(s_i)} \rangle \! \langle b_{o_i}^{(s_i)} |$. 
Let $\mathcal{M} = \bigcup_{\vec{s}} \mathcal{M}^{(\vec{s})}: \mathbb{H}_d \to \mathbb{R}^{3^k} \times \mathbb{R}^{2^k}$ denote the union of all such basis measurements.
Then, the following formula is true for tensor product matrices $X = \bigotimes_{i=1}^k X_i$ and $X_i \in \mathbb{H}_2$:
\begin{align}
\mathcal{M}^\dagger \mathcal{M}(X)&=  \sum_{\vec{s},\vec{o}} \langle b^{(\vec{s})}_{\vec{o}}| X |b^{(\vec{s})}_{\vec{o}} \rangle | b^{(\vec{s})}_{\vec{o}} \rangle \! \langle b^{(\vec{s})}_{\vec{o}}| \nonumber \\
=& \bigotimes_{i=1}^k \left( \sum_{s_i,o_i} \langle b^{(s_i)}_{o_i} | X_i | b^{(s_i)}_{o_i} \rangle |b^{(s_i)}_{o_i} \rangle \! \langle b^{(s_i)}_{o_i}| \right) \nonumber \\
=& 3^k\bigotimes_{i=1}^k \mathcal{D}(X_i ) 
= 3^k\mathcal{D}^{\otimes k}_{1/3} (X), \label{eq:pauli_basis_aux2}
\end{align}
where we have used Eq.\eqref{eq:single_qubit}. Linear extension ensures that this formula remains valid for arbitrary matrices $X \in \mathbb{H}_d$. Since the single qubit depolarizing channel is invertible, the same is true for its $k$-fold tensor product and we conclude
\begin{align*}
\left( \mathcal{M}^\dagger \mathcal{M} \right)^{-1} (X) = \frac{1}{3^k} \left( \mathcal{D}_{1/3}^{\otimes k} \right)^{-1} (X).
\end{align*}
Inserting this explicit expression into the closed-form expression for the least squares estimator yields
\begin{align*}
\hat{L}_n =& \left( \mathcal{M}^\dagger \mathcal{M} \right)^{-1} \left( \mathcal{M}^\dagger (\vec{f}) \right) \\
=& \frac{1}{3^k} \sum_{\vec{s},\vec{o}} \left[ f \right]^{(\vec{s})}_{\vec{o}} \left( \mathcal{D}^{\otimes k}_{1/3} \right)^{-1} \left( |b^{(\vec{s})}_{\vec{o}} \rangle \! \langle b^{(\vec{s})}_{\vec{o}}| \right),
\end{align*}
as advertised in the main text. Here, $\left[ f \right]_{\vec{o}}^{(\vec{s})}$ is assumed to be a frequency approximation to $p_o^{\vec{(s)}} = \langle b^{(\vec{s})}_{\vec{o}} | \rho | b^{(\vec{s})}_{\vec{o}} \rangle$.

We conclude this section with a single-qubit observations that allows for characterizing this expression in a more explicit fashion.
Note that one may rewrite $\mathcal{D}_{1/3}(X)$ as $\frac{\mathrm{tr}(X)}{2} \mathbb{I} + \frac{1}{6} \sum_s \mathrm{tr} \left( \sigma_s X \right) \sigma_s$. This facilitates the computation of the single-qubit inverse:
\begin{equation*}
\mathcal{D}_{1/3}^{-1} \left( X \right)
= \frac{\mathrm{tr}(X)}{2} \mathbb{I} + \frac{3}{2} \sum_s \mathrm{tr} \left( \sigma_s  X \right)\sigma_s
\end{equation*}
and, in particular
\begin{align*}
\mathcal{D}^{-1}_{1/3} \left( |b^{(s)}_{\pm} \rangle \! \langle b^{(s)}_{\pm}| \right) 
=& \frac{1}{2}\left( \mathcal{D}^{-1}(\mathbb{I}) \pm \mathcal{D}^{-1} (\sigma_s ) \right) \nonumber \\
=& \frac{1}{2} \left( \mathbb{I} \pm 3 \sigma_s \right)
= 3  |b^{(s)}_\pm \rangle \! \langle b^{(s)}_\pm| - \mathbb{I}. 
\end{align*}
This in turn implies
\begin{align}
\hat{L}_n = \frac{1}{3^k} \sum_{\vec{s},\vec{o}} \left[ f \right]_{\vec{o}}^{(\vec{s})} \bigotimes_{i=1}^k \left( 3 |b_{o_i}^{(s_i)} \rangle \! \langle b_{o_i}^{(s_i)} | - \mathbb{I} \right). \label{eq:LI_pauli_basis2_appendix}
\end{align}
which, again, is an unbiased estimator with respect to random fluctuations in the sample statistics.

Finally, we also point out another consequence that will be important later on: 
\begin{align}
\left( \mathcal{D}^{-1}_{1/3} \left( |b^{(s)}_{o} \rangle \! \langle b^{(s)}_{o}| \right)  \right)^2
=& 5 \mathcal{D}_{3/5} \left( |b_o^{(s)} \rangle \! \langle b_o^{(s)}| \right),  \label{eq:pauli_basis_aux1}
\end{align}
where $\mathcal{D}_{3/5}$ is another single-qubit depolarizing channel.

\section{the matrix Bernstein inequality and concentration in operator norm} \label{sec:concentration}

Scalar concentration inequalities provide sharp bounds on the probability of a sum of independent random variables deviating from their mean. 
Classical examples include Hoeffding's, Chernoff's and Bernstein's inequality -- all of which have found widespread use in a variety of scientific disciplines. 
The main results of this work are based on a matrix generalizations of these classical statements -- in particular the \emph{matrix Bernstein inequality} developed by one of the authors, see \cite[Theorem~1.4]{tropp_user-friendly_2012}.

\begin{theorem}[Matrix Bernstein inequality] \label{thm:bernstein}
Consider a sequence of $n$ independent, hermitian random matrices $A_1,\ldots,A_n \in \mathbb{H}_d$. Assume that each $A_i$ satisfies
\begin{align*}
\mathbb{E} \left[ A_i \right] = 0 \quad \textrm{and} \quad \| A_i \|_\infty \leq R\textrm{ almost surely}.
\end{align*}
Then,  for any $t >0$
\begin{equation*}
\mathrm{Pr} \left[ \left\| \sum_{i=1}^n \left( A_i - \mathbb{E} \left[ A_i \right] \right) \right\|_\infty \geq t \right] \leq 
\begin{cases}
d \exp \left( - \frac{3 t^2}{8 \sigma^2} \right) & t \leq \frac{\sigma^2}{R}, \\
d \exp \left( - \frac{3t}{8R} \right) & t \geq \frac{\sigma^2}{R},
\end{cases}
\end{equation*}
where $\sigma^2 = \left\| \sum_{i=1}^n \mathbb{E} \left[ A_i^2 \right] \right\|_\infty$.
\end{theorem}

First results of this kind originate in Banach space theory \cite{tomczak_moduli_1974,lust_inegalites1986,pisier_noncommutative_1997,rudelson_random_1999} and were later independently developed in quantum information theory \cite{ahlswede_strong_2002,gross_recovering_2011}.
Further advances by Oliveira \cite{oliveira_sums_2010} and and one of the authors \cite{tropp_user-friendly_2012} led to the result that we employ here. We refer to the monograph \cite{tropp_user-friendly_2012} for a detailed exposition of related work and history.

Similar to the scalar Bernstein inequality, the tail behavior in Theorem~\ref{thm:bernstein} consists of two regimes. Small deviations are suppressed in a subgaussian fashion, while larger deviations follow a subexponential decay. The ratio $\frac{\sigma^2}{R}$ marks the transition from one regime into the other. 
We also note in passing that this result recovers the scalar Bernstein inequality for $d=1$ ($\mathbb{H}_1 \simeq \mathbb{R}$).

\subsection{Concentration for structured POVM measurements}
\label{sub:concentration_structured}

For structured measurements (2-designs) we may rewrite the (plain) least squares estimator \eqref{eq:LI_2design_appendix} as
\begin{align*}
\hat{L}_n
=& (d+1) \sum_{i=1}^m \left[f_n \right]_i |v_i \rangle \! \langle v_i| - \mathbb{I} = \frac{1}{n} \sum_{i=1}^m X_i,
\end{align*}
where each $X_i$ is an i.i.d. copy of the random matrix $X \in \mathbb{H}_d$ that assumes $(d+1)|v_k \rangle \! \langle v_k|-\mathbb{I}$ with probability $\frac{d}{m}\langle v_k| \rho |v_k \rangle$ for all $k \in \left[m \right]$.
Unbiasedness with respect to the sample statistics ensures $\mathbb{E} \left[ \hat{L}_n \right] = \mathbb{E} \left[ X \right] = \rho$. Hence, $\hat{L}_n - \rho$ is a sum of iid, centered random matrices $\frac{1}{n}\left(X_i - \mathbb{E} \left[ X_i \right]\right)$. These obey
\begin{align*}
\| X_i - \mathbb{E} \left[ X_i \right] \|_\infty = &\frac{1}{n}\left\|(d+1) |v_k \rangle \! \langle v_k| - \mathbb{I} -\rho \right\|_\infty
\leq \frac{d}{n} =:R,
\end{align*}
where $k \in \left[ m \right]$ is arbitrary. 
Next, note that the random matrix $X$ obeys
\begin{align*}
\mathbb{E} \left[ \left(X - \mathbb{E} \left[ X \right] \right)^2 \right] 
=& \mathbb{E} \left[ X^2 \right] - \mathbb{E} \left[ X \right]^2
= \mathbb{E} \left[ X^2 \right] - \rho^2
\end{align*}
and also
\begin{align*}
 \mathbb{E} \left[ X^2 \right] 
=& \sum_{k=1}^m \frac{d}{m} \langle v_k| \rho |v_k \rangle \left( (d+1) |v_k \rangle \! \langle v_k| - \mathbb{I} \right)^2 \\
=& \frac{d(d^2-1)}{N} \sum_{k=1}^N \langle v_k| \rho |v_k \rangle |v_k \rangle \! \langle v_k| + \mathbb{I} \\
=& (d-1) \left( \rho + \mathbb{I} \right) + \mathbb{I},
\end{align*}
according to Eq.~\eqref{eq:2design_appendix}. This allows us to bound the variance parameter:
\begin{align*}
\left\|\frac{1}{n}\sum_{i=1}^n \left( \frac{1}{n}(X - \mathbb{E} \left[ X \right] \right)^2 \right\|_\infty
=& \frac{1}{n} \left\| (d+1) \rho + d \mathbb{I} -\rho^2 \right\|_\infty  \\
\leq & \frac{2d}{n} =: \sigma^2.
\end{align*}
The ratio $\frac{\sigma^2}{R} = 2$ indicates that any choice of $\tau \in \left[0,2 \right]$ will fall into the subgaussian regime of the matrix Bernstein inequality and  Theorem~\ref{thm:bernstein} yields 
\begin{align}
\mathrm{Pr} \left[ \left\| \hat{L}_n-\rho \right\|_\infty \geq \tau \right]
=& \mathrm{Pr} \left[ \left\| \frac{1}{n} \sum_{i=1}^n \left( X_i - \mathbb{E} \left[ X_i \right] \right) \right\|_\infty \geq \tau \right] \nonumber \\
\leq &
 d \mathrm{e}^{ - \frac{3 \tau^2 n}{16 d}}.
\label{eq:closeness_2design}
\end{align}

\subsection{Concentration for (global) Pauli observables} \label{sec:paulis1}

We assume that the total number of samples $n$ is distributed equally among the $d^2$ different Pauli measurements.
Similar to before, unbiasedness \eqref{eq:pauli_unbiased} and the explicit characterization of the LI estimator \eqref{eq:LI_pauli_appendix} allow us to write
\begin{align*}
\hat{L}_n -\rho  
= \sum_{k=1}^{d^2} \frac{1}{n} \sum_{i=1}^{n/d^2} \left( X_i^{(k)}- \mathbb{E} \left[ X_i^{(k)} \right] \right).
\end{align*}
Here, each $X_i^{(k)}$ is an independent instance of the random matrix $X^{(k)} = \pm d W_k$ with probability $\frac{1}{2} \left( 1 \pm \mathrm{tr}(W_k \rho) \right)$ each. 
This is a sum of centered random matrices that are independent, but in general not identically distributed. However, independence alone suffices for applying Theorem~\ref{thm:bernstein}.
We note in passing that this would not be the case for earlier (weaker) versions of the matrix Bernstein inequality.
Bound
\begin{equation*}
\frac{1}{n}\left\| X_i^{(k)} - \mathbb{E} \left[ X_i^{(k)} \right] \right\|_\infty = \frac{d}{n} \left\| \left(1 \pm \mathrm{tr} \left( W_i \rho \right) \right) \mathbb{I} \right\|_\infty \leq  \frac{2d}{n} =:R
\end{equation*}
and use the fact that
\begin{align}
\mathbb{E} \left[ \left( X_i^{(k)} - \mathbb{E} \left[ X_i^{(k)} \right]\right)^2 \right]
=& \mathbb{E} \left[ \left( X^{(k)} \right)^2 \right] - \mathbb{E} \left[ X^{(k)} \right]^2 \nonumber \\
 \leq & \mathbb{E} \left[ \left( X^{(k)} \right)^2 \right]  \label{eq:variance_aux}
\end{align}
(in the positive semidefinite order) to considerably simplify the variance computation:
\begin{align*}
 \frac{d^2}{n} \left\| \sum_{k=1}^{d^2} \mathbb{E} \left[ \left( X^{(k)} \right)^2 \right] \right\|_\infty 
=& \frac{1}{n} \left\| \sum_{k=1}^{d^2} \mathbb{I} \right\|_\infty = \frac{d^2}{n} =: \sigma^2,
\end{align*}
because $\left( X^{(k)} \right)^2 = \frac{1}{d^2} \mathbb{I}$.
Applying Theorem~\ref{thm:bernstein} yields
\begin{align*}
\mathrm{Pr} \left[ \left\| \hat{L}_n - \rho \right\|_\infty \geq \tau \right] \leq d \mathrm{e}^{-\frac{3 \tau^2 \tilde{n}}{8}} \quad \tau \in \left[0, d/2 \right].
\end{align*}

\subsection{Concentration for Pauli-basis measurements}

Once more, we assume that the total budget of samples $n$ is distributed equally among all $3^k$ Pauli basis choices. 
Unbiasedness of the LI estimator together with the explicit description \eqref{eq:LI_pauli_basis2_appendix} allows us to once more interpret $\hat{L}_n - \rho$ as a sum of independent, centered random matrices:
\begin{equation*}
\hat{L}_n - \rho 
= \sum_{\vec{s}} \frac{1}{n} \sum_{i=1}^{n/3^k} \left( X_i^{(\vec{s})}- \mathbb{E} \left[ X^{(\vec{s})}_i \right] \right).
\end{equation*}
For each $\vec{s} \in \left\{x,y,z \right\}^k$, $X_i^{(\vec{s})}$ is an independent copy of the random matrix
\begin{align*}
X^{(\vec{s})} = \bigotimes_{i=1}^k \left( 3 |b_{o_i}^{(s_i)} \rangle \! \langle  b_{o_i}^{(s_i)}| - \mathbb{I} \right)  
\end{align*}
with probability $\langle b^{(\vec{s})}_{\vec{o}}| \rho |b^{(\vec{s})}_{\vec{o}} \rangle$ for each $\vec{o} \in \left\{ \pm 1 \right\}^k$.
Jensen's inequality implies
\begin{align*}
&\frac{1}{n} \left\| X^{(\vec{s})} - \mathbb{E} \left[ X^{(\vec{s})} \right] \right\|_\infty \leq  \frac{2}{n} \left\| X^{(\vec{s})} \right\|_\infty \\
=& \frac{2}{n} \prod_{i=1}^k \left\|3 |b_{o_i}^{(s_i)} \rangle \! \langle b_{o_i}^{(s_i)}-\mathbb{I} \right\|_\infty = \frac{2^{k+1}}{n}=:R.
\end{align*}
For the variance, we once more use \eqref{eq:variance_aux}
and compute
\begin{align*}
 \frac{1}{n} \sum_{\vec{s}} \mathbb{E}  \left( X^{(\vec{s})} \right)^2 
=& \sum_{\vec{s},\vec{o}}\langle b^{(\vec{s})}_{\vec{o}} | \rho | b^{(\vec{s})}_{\vec{o}} \rangle \bigotimes_{i=1}^k \left( \mathcal{D}^{-1} \left( |b^{(\vec{s})}_{\vec{o}} \rangle \! \langle b^{(\vec{s})}_{\vec{o}} | \right)\right)^2 \\
=&  \sum_{\vec{s}}\sum_{\vec{o}} \langle b^{(\vec{s})}_{\vec{o}} | \rho | b^{(\vec{s})}_{\vec{o}} \rangle \bigotimes_{i=1}^k \frac{5}{3} \mathcal{D}_{3/5} \left( |b^{(\vec{s})}_{\vec{o}} \rangle \! \langle b^{(\vec{s})}_{\vec{o}} | \right) \\
=& 5^k\mathcal{D}_{3/5}^{\otimes k} \left( \frac{1}{3^k} \sum_{\vec{s},\vec{o}} \langle b_{\vec{o}}^{(\vec{s})} | \rho |b^{(\vec{s})}_{\vec{o}}\rangle |b_{\vec{o}}^{(\vec{s})} \rangle \! \langle b_{\vec{o}}^{(\vec{s})} |\right) \\
=& 5^k\mathcal{D}^{\otimes k}_{3/5} \left( \mathcal{D}_{1/3}^{\otimes k} \left( \rho \right) \right) 
= \frac{5^k}{n} \mathcal{D}^{\otimes k}_{1/5} (\rho),
\end{align*}
where we have used Eq.~\eqref{eq:pauli_basis_aux1} and the fact that the combination of two depolarizing channels is again a depolarizing channel.
This expression can be evaluated explicitly. For $\alpha \subset \left[ k \right]$, let $\mathrm{tr}_{\alpha}(\rho)$ denote the partial trace over all indices contained in $\alpha$. Then, for $X=\bigotimes_{j=1}^k X_j\in \mathbb{H}_d$
\begin{align*}
5^k \mathcal{D}_{1/5}^{\otimes k}\left( \bigotimes_{j=1}^k X_j \right) =& \bigotimes_{j=1}^k \left( \mathrm{tr}(X_j) \mathbb{I} + X_j \right)  \\
=& \sum_{\alpha \subset \left[k \right]} 2^{|\alpha|} \mathrm{tr}_\alpha(X) \otimes \mathbb{I}^{\otimes \alpha} 
\end{align*}
and this extends linearly to all of $\mathbb{H}_d \simeq\mathbb{H}_2^{\otimes k}$. Consequently,
\begin{align*}
\left\| \frac{1}{n} \sum_{\vec{s}} \mathbb{E} \left[ \left( X^{(\vec{s})} \right)^2 \right] \right\|_\infty 
=& \frac{1}{n}\left\| \sum_{\alpha \subset \left[ k \right]} 2^{|\alpha|} \mathrm{tr}_\alpha (\rho) \otimes \mathbb{I}^{\otimes \alpha} \right\|_\infty \\
\leq & \frac{1}{n}\sum_{\alpha \subset \left[ k \right]} 2^{|\alpha|} \| \mathrm{tr}_\alpha (\rho) \|_\infty \| \mathbb{I}^{\otimes \alpha} \|_\infty \\
\leq & \frac{1}{n} \sum_{\alpha \subset \left[ k \right]} 2^{| \alpha|}
= \frac{1}{n}\sum_{j=0}^k \binom{k}{j}2^j \\
=& \frac{(2+1)^k}{n} = \frac{3^k}{n} =: \sigma^2.
\end{align*}
This estimate is actually tight for pure product states of the form $\rho = (|\psi \rangle \! \langle \psi|)^{\otimes k}$.
We may now apply Theorem~\ref{thm:bernstein} to conclude
\begin{equation*}
\mathrm{Pr} \left[ \left\| \hat{L}_n - \rho \right\|_\infty \geq \tau \right] \leq d \exp \left( - \frac{3 n \tau^2}{8 \times 3^k} \right)
\quad \tau \in \left[0,1 \right].
\end{equation*}

\section{Conversion of confidence regions from operator norm to trace norm} \label{sec:conversion}

The final ingredient for the framework presented in this manuscript is a reliable way to transform operator-norm closeness of the (plain) least squares estimator $\hat{L}_n$
into a statement about closeness of the PLS estimator  $\hat{\rho}_n$
in trace distance. Recall that the optimization problem \eqref{eq:pli_appendix} admits an analytic solution \cite{smolin_efficient_2012}. Let $U \mathrm{diag}(\vec{\lambda}) U^\dagger$ be an eigenvalue decomposition of $\hat{L}_n$. Then,
\begin{equation}
\hat{\rho}_n = U \mathrm{diag}\left(\left[ \vec{\lambda} - x_0 \vec{1} \right]^+\right) U^\dagger,
\label{eq:pli_analytic}
\end{equation}
where $x_0$ is chosen such that $\mathrm{tr}(\hat{\rho}_n) =1$ and $\left[ \vec{y} \right]_i^+ = \max \left\{ \left[ \vec{y} \right]_i, 0 \right\}$ denotes thresholding on non-negative components. This solution is \emph{unique}, provided that $\mathrm{tr} \left( \hat{L}_n \right) =1$, which is the case for all the least squares estimators we consider.

The conversion from closeness in operator norm to closeness in trace norm will introduce a factor that is proportional to the \emph{effective rank} of the density matrix $\rho$, rather than a full dimensional factor. For $r \in \mathbb{N}$, we define the best rank-$r$ approximation $\rho_r$ of a quantum state  $\rho \in \mathbb{H}_d$ as the optimal feasible point of 
\begin{equation}
\sigma_r (\rho) =\underset{\mathrm{rank}(Z) \leq r}{\textrm{minimize}}  \quad \| \rho-Z \|_1. \label{eq:approximation_error}
\end{equation}
This problem can be solved analytically.
Let $\rho = \sum_{i=1}^d \lambda_i |x_i \rangle \! \langle x_i|$ be an eigenvalue decomposition with eigenvalues arranged in non-increasing order. Then,
\begin{align*}
Z^\sharp =& \sum_{i=1}^r \lambda_i |x_i \rangle \! \langle x_i|, \quad \textrm{and} \quad 
\sigma_r (\rho) = \sum_{i=r+1}^d \lambda_i = 1 - \mathrm{tr}(\rho_r),
\end{align*}
highlighting that the best rank-$r$ approximation is simply a truncation onto the $r$ largest contributions in the eigenvalue decomposition.
This truncated description is accurate if the residual error $\sigma_r (\rho)$ is small. If this is the case it is reasonable to say that $\rho$ is well approximated by a rank-$r$ matrix $Z^\sharp$ and has \emph{effective rank $r$}.

\begin{proposition} \label{prop:trace_distance}
Suppose that $\hat{L}_n \in \mathbb{H}_d$ obeys $\mathrm{tr} \left( \hat{L}_n \right) = 1$ and $\| \hat{L}_n - \rho \|_\infty \leq \tau$ for some quantum state $\rho \in \mathbb{H}_d$ and $\tau \geq 0$. 
Then, for any $r \in \mathbb{N}$ the PLS estimator $\hat{\rho}_n$ obeys
\begin{equation*}
\left\| \hat{\rho}_n - \rho \right\|_1 \leq 4 r \tau + 2 \min \left\{ \sigma_r (\rho), \sigma_r \left( \hat{\rho}_n \right) \right\},
\end{equation*}
where $\sigma_r (\rho)$ is defined in Eq.~\eqref{eq:approximation_error}.
\end{proposition}

The statement readily follows from combining two auxiliary results. The first one states that the threshold value $x_0$ in the analytic solution of $\hat{\rho}_n$ must be small if $\hat{L}_n$ is operator-norm close to a quantum state.

\begin{lemma} \label{lem:conversion_aux1}
Instantiate the assumptions from Proposition~\ref{prop:trace_distance}. Then, the threshold value in Eq.~\eqref{eq:pli_analytic} obeys $x_0 \in \left[0, \tau \right]$.
\end{lemma}

\begin{proof}
By assumption $\hat{L}_n$ has unit trace. If it is in addition psd, $\hat{\rho}_n = \hat{L}_n$, because $\hat{L}_n$ is already a quantum state and the projection is trivial ($x_0=0$)
Otherwise, $\hat{L}_n$ is indefinite and unit trace ensures that the positive part dominates. Hence, $x_0$ must be strictly positive to enforce $\mathrm{tr}(\hat{\rho}_n)=1$.

For the upper bound, let $P \in \mathbb{H}_d$ denote the orthogonal projection onto the range of $\hat{\rho}_n$. Then, $\hat{\rho}_n = P \left( \hat{L}_n- x_0 \mathbb{I} \right) P = P \hat{L}_n P - x_0 P$, according to Eq.~\eqref{eq:pli_analytic}. 
In semidefinite order, this implies
\begin{align*}
\hat{\rho}_n=& P (\hat{L}_n-\rho )P - x_0 P + P \rho P \\
\leq & \left( \| \hat{L}_n - \rho \|_\infty - x_0 \right) P + P \rho P \\
\leq &  \left( \tau - x_0 \right) P + P \rho P,
\end{align*}
where the last line follows from the assumption $\| \hat{L}_n - \rho \|_\infty \leq \tau$. The trace preserves semidefinite order and we conclude
\begin{align*}
0 \leq \mathrm{tr} \left( \hat{\rho}_n\right) - \mathrm{tr} \left( P \rho P \right)
\leq \left(\tau - x_0 \right) \mathrm{tr} \left( P \right) 
\end{align*}
which implies an upper bound of $\tau$ ($\mathrm{tr}(P)>0$).
\end{proof}

The second technical lemma generalizes a result that is somewhat folklore in quantum information theory: the ``effective rank'' of a difference of two quantum states is proportional to the minimal rank of the two density operators involved.

\begin{lemma} \label{lem:conversion_aux2}
Fix $r \in \mathbb{N}$ and let $\rho, \sigma \in \mathbb{H}_d$ be quantum states. Then,
\begin{equation*}
\| \rho - \sigma \|_1 \leq 2 r \| \rho - \sigma \|_\infty + 2 \min \left\{ \sigma_r (\rho), \sigma_r (\sigma)\right\},
\end{equation*}
where the residual error $\sigma_r (\cdot)$ was defined in Eq.~\eqref{eq:approximation_error}.
\end{lemma}

\begin{proof}
We can without loss of generality assume $\sigma_r (\rho) \leq \sigma_r (\sigma)$. 
Decompose $\rho$ into $\rho_r + \rho_c$, where $\rho_r$ is the best rank-$r$ approximation \eqref{eq:approximation_error} and $\rho_c = \rho-\rho_r$ denotes the ``tail''. By construction, both $\rho_r$ and $\rho_c$ are positive semidefinite matrices that obey $\sigma_r (\rho) = \mathrm{tr}(\rho_c) = 1- \mathrm{tr}(\rho_r)$. 
The triangle inequality then implies
\begin{equation*}
\| \rho - \sigma \|_1 \leq \| \rho_r - \sigma \|_1 + \sigma_r (\rho),
\end{equation*}
because $\| \rho_c \|_1 = \sigma_r (\rho)$. 
Next, let $P_+,P_- \in \mathbb{H}_d$ be the projections onto the positive and non-positive ranges of $\rho_r - \sigma$. 
By construction,  $P_+$ has rank at most $\mathrm{rank}(\rho_r)=r$ and the trace norm equals
\begin{align*}
\| \rho_r - \sigma \|_1 = \mathrm{tr} \left(P_+ (\rho_r-\sigma) \right) - \mathrm{tr} \left( P_- (\rho_r-\sigma) \right).
\end{align*}
On the other hand,
\begin{equation*}
\sigma_r (\rho) = \mathrm{tr}(\sigma-\rho_r) = -\mathrm{tr} \left( P_+ (\rho_r - \sigma) \right) - \mathrm{tr} \left( P_- (\rho_r-\sigma) \right),
\end{equation*}
because $P_+ + P_- = \mathbb{I}$.
Combining both relations yields
\begin{align*}
\| \rho_r - \sigma \|_1 =& 2 \mathrm{tr}\left(P_+ (\rho_r-\sigma) \right) + \sigma_r (\rho) \\ 
\leq & 2 \mathrm{tr} \left( P_+ (\rho-\sigma) \right) + \sigma_r (\rho) \\
\leq & 2 \| P_+ \|_\infty \| \rho - \sigma \|_1 + \sigma_r (\rho),
\end{align*}
where we have used $\mathrm{tr}(P_+ \rho_c) \geq 0$ and Hoelder's inequality. Finally, note that $\| P_+ \|_1 = \mathrm{rank}(P_+)=r$ by construction and the claim follows.
\end{proof}

The main result of this section is a rather straightforward combination of these two technical statements.

\begin{proof}[Proof of Proposition~\ref{prop:trace_distance}]
Fix $r \in \mathbb{N}$ and use, Lemma~\ref{lem:conversion_aux2} to conclude
\begin{align*}
\| \hat{\rho}_n - \rho \|_1 \leq 2 r \| \hat{\rho}_n - \rho \|_\infty +2 \min \left\{ \sigma_r (\rho), \sigma_r \left( \hat{\rho}_n \right) \right\}.
\end{align*}
Next, note that according to \eqref{eq:pli_analytic}, $\hat{\rho}_n$ may be viewed as the positive definite part of the matrix $\hat{L}_n - x_0 \mathbb{I}$. Such a restriction to the positive part can never increase the operator norm distance to another positive semidefinite matrix. Hence,
\begin{align*}
\| \hat{\rho}_n- \rho \|_\infty
\leq \| \hat{L}_n - \rho \|_\infty + |x_0| \| \mathbb{I} \|_\infty 
\leq 2 \tau,
\end{align*}
where the last inequality follows from Lemma~\ref{lem:conversion_aux1}.
\end{proof}

\section{Proof of the main result}
\label{sec:effective_rank}

By now we have everything in place to provide a complete proof of the main result of this work.

\begin{theorem} \label{thm:main_appendix}
Let $\rho \in \mathbb{N}$ be a state. Suppose that we either perform $n$ structured POVM measurements (set $g(d) = 2d$), $n$ Pauli observable measurements (set $g(d) = d^2$), or  $n$ Pauli basis measurements (set $g(d) =d^{1.6}$).
Then, for any $r \in \mathbb{N}$ and $\epsilon \in \left[0,1 \right]$,the PLS estimator $\hat{\rho}_n$ \eqref{eq:pli_appendix} obeys
\begin{align*}
\mathrm{Pr} \left[ \left\| \hat{\rho}_n - \rho \right\|_1 \geq \epsilon + 2 \min \left\{\sigma_r (\rho), \sigma_r (\hat{\rho}_n) \right\} \right] \leq d \mathrm{e}^{-\frac{n \epsilon^2}{43 g(d)r^2}},
\end{align*}
where $\sigma_r (\rho),\sigma_r (\hat{\rho}_n)$ denote the residual error of approximating $\rho$ and $\hat{\rho}_n$ by a rank-$r$ matrix \eqref{eq:approximation_error}.
\end{theorem}

Note that Theorem~\ref{thm:main_result} is an immediate consequence of this more general result: simply set $r = \min \left\{ \mathrm{rank}(\rho), \mathrm{rank}(\hat{\rho}_n) \right\}$ which in turn ensures $\min \left\{\sigma_r (\rho), \sigma_r (\hat{\rho}_n)\right\} = 0$.

However, unlike this specification, Theorem~\ref{thm:main_appendix} does feature an additional degree of freedom. The parameter $r \in \mathbb{N}$ allows for interpolating between small values (small sampling rate, but a potentially large reconstruction error) and large values (high sampling rate, but low reconstruction error).  
This tradeoff is particulary benign for quantum states that are approximately low-rank.
Due to experimental imperfections, such states arise naturally in many experiments that aim at generating a  pure quantum state.
We illustrate this by means of the following caricature of a faulty state preparation protocol.
Suppose that an apparatus either produces a target state $| \psi \rangle \! \langle \psi|$ perfectly, or fails completely,
 in the sense that it outputs a maximally mixed state. Then, the resulting state is
\begin{equation*}
\rho = (1-p)| \psi \rangle \! \langle \psi| + \frac{p}{d}\mathbb{I},
\end{equation*} 
where $p \in \left[0,1 \right]$ denotes the probability of failure.
This state has clearly full rank and Corollary~\ref{cor:error_bars} requires at least $n \geq 43 \frac{g(d)d^2}{\epsilon^2} \log (d/\delta)$ samples to estimate it up to trace-norm accuracy $\epsilon$ with high probability. 
In contrast, Theorem~\ref{thm:main_appendix} ensures that already
$
n \geq 43 \frac{g(d)}{\epsilon^2} \log (d/\delta)
$
samples suffice to ensure that, with high probability, the PLS estimator obeys
$
\| \hat{\rho}_n - \rho \|_1 \leq \epsilon+2p
$.
For sufficiently high success probabilities/low accuracy ($\epsilon \geq 2p/d$) this clearly outperforms the original statement.

\begin{proof}[Proof of Theorem~\ref{thm:main_appendix}]

We illustrate the proof for structured POVMs -- the other settings are completely analogous. 
Fix $\epsilon \in \left[0,1 \right]$, $r \in \mathbb{N}$ and set $\tau = \frac{\epsilon}{4 r}$. Then, the main result of Sec.~\ref{sub:concentration_structured} -- Equation~\eqref{eq:closeness_2design} -- ensures that the least squares estimator $\hat{L}_n$ obeys
\begin{equation}
\left\| \hat{L}_n - \rho \right\|_\infty \leq \tau = \frac{\epsilon}{4 r}
\label{eq:main_appendix_aux1}
\end{equation}
with probability of failure bounded by
$
d \mathrm{e}^{-\frac{\epsilon^2 n}{86 dr^2}}
$. Assuming that this condition is true, Proposition~\ref{prop:trace_distance} readily yields
$
\| \hat{\rho}_n - \rho \|_1 \leq \epsilon + 2 \min \left\{ \sigma_r (\rho), \sigma_r (\hat{\rho}_n) \right\}
$.
\end{proof}

\section{Improved convergence guarantees for the uniform POVM}
\label{sec:uniform}

All the convergence results derived so far feature an additional $\log(d)$-factor. This is a consequence of the matrix Bernstein inequality that proved instrumental in deriving these results. One can show that such an additional factor necessarily features in \emph{all} matrix concentration inequalities that are based on exclusively first and second moments of the random matrices in question \cite{tropp_second-order_2018}.

However, the following question remains: is this $\log (d)$-factor in Theorem~\ref{thm:main_result} an artifact of the proof technique, or is it an intrinsic feature of tomography via projected least squares?

In this section we rule out the second possibility: a different proof technique allows for avoiding this $\log (d)$-factor, provided that the POVM is sufficiently symmetric and well-behaved. More precisely, we re-visit the uniform POVM $\left\{ d |v \rangle \! \langle v| \mathrm{d}v \right\}_{v \in\mathbb{S}^d}$ and exploit the fact that Eq.~\eqref{eq:frame_operator} completely characterizes \emph{all} moments of the resulting outcome distribution. 
This opens the door for applying very strong proof techniques from large dimensional probability theory that found wide-spread applications in a variety of subjects, including compressed sensing \cite{foucart_mathematical_2013} and, more recently, quantum information theory \cite{lancien_approximating_2017}. 
We believe that this technique may be of independent interest and find it therefore worthwhile to present it in a self-contained fashion. 
Roughly speaking, it is based on the following steps:
\begin{itemize}
\item[(o)] \emph{Reformulation:} the operator norm of a random hermitian matrix $A \in \mathbb{H}_d$ admits a variational definition: 
\begin{equation}
\| A \|_\infty = \max_{y \in \mathbb{S}^d} \left| \langle y| A |y \rangle \right|.
\label{eq:variational_operator_norm}
\end{equation} 
\item[(i)] \emph{Discretization:} replace the maximization over the entire complex unit sphere $\mathbb{S}^d$ by a maximization over a finite point set $\mathcal{N}$ that covers $\mathbb{S}^d$  to sufficiently high accuracy (\emph{covering net}).
\item[(ii)] \emph{Concentration:} fix $y \in \mathcal{N}$ and show that the scalar random variable $s_y=\langle y| A |y \rangle$ concentrates sharply around its expectation value.
\item[(iii)] \emph{Union bound:} apply a union bound over all $|\mathcal{N}|$ random variables $s_y$ to obtain an upper bound on the operator norm $\| A \|_\infty$.
\end{itemize}
Ideally the tail bound from (iii) is sharp enough to ''counter-balance'' the $|\mathcal{N}|$-pre-factor that results from the union bound in step (iv). Should this be not the case, more sophisticated methods, like generic chaining \cite{talagrand_generic_2006}, may still allow for drawing non-trivial conclusions. Fortunately, for the task at hand, this turns out to not be necessary and the rather naive strategy sketched above suffices to achieve a result that is (provably) optimal up to a constant factors:

\begin{theorem} \label{thm:uniform}
Suppose that we perform $n$ independent uniform POVM measurements on a quantum state $\rho \in \mathbb{H}_d$. Then, the associated least squares estimator $\hat{L}_n$ obeys
\begin{equation*}
\mathrm{Pr} \left[ \left\| \hat{L}_n - \rho \right\|_\infty \geq \tau \right]
\leq 2 \exp \left( c_1 d - c_2 n\tau^2 \right),
\end{equation*}
In particular, $n \geq C \frac{d}{\tau^2} \log (1/\delta)$ suffices to ensure $\| \hat{L}_n - \rho \|_\infty \leq \tau$ with probability at least $1-\delta$.
Here $c_1,c_2,C >0$ denote constants of sufficient size.
\end{theorem}

No effort has been made to optimize the constants. The proof presented here yields $c_1 = 2 \log (3)$ and $c_2 = \frac{1}{480}$
which could be further improved by a more careful analysis.
Importantly, the second part of this statement can be combined with Proposition~\ref{prop:trace_distance} to readily deduce the last technical result of the main text:

\begin{corollary}[Re-statement of Theorem~\ref{thm:uniform}]
For any rank-$r$ state $\rho$, a number of $n \geq C \frac{r^2d}{\epsilon^2}\log (1/\delta)$ uniform POVM measurements suffice to ensure $\| \hat{\rho}^\sharp_{(n)} - \rho \|_1 \leq \epsilon$ with probability at least $1-\delta$.
\end{corollary}

Not only does this statement reproduce the best known sampling rates for tomography with independent measurements \cite{kueng_low_2017}, it also exactly matches lower bounds on the minimal sample complexity associated with \emph{any} tomographic procedure that may apply in this setting \cite[Table~I]{haah_sample_2017}.

The remainder of this section is dedicated to proving Theorem~\ref{thm:uniform}. 
For the sake of accessibility, we will divide this proof into three subsections that contain the steps summarized above.

\subsection{Step I: Reformulation and discretization}

Suppose that we perform $n$ uniform POVM measurements on a fixed quantum state $\rho \in \mathbb{H}_d$. Then, the least squares estimator is equivalent to a sum of i.i.d. random matrices:
\begin{equation*}
\hat{L}_n = \frac{1}{n} \sum_{i=1}^n X_i.
\end{equation*}
Each $X_i$ is an independent copy of the random matrix $X$ that assumes the value $(d+1) |v \rangle \! \langle v| - \mathbb{I}$ with probability $d \langle v| \rho |v \rangle \mathrm{d}v$ and $v$ may range over the entire complex unit sphere.
Unbiasedness of this estimator in turn implies
\begin{align*}
\| \hat{L}_n - \rho \|_\infty
=& \left\| \frac{1}{n} \sum_{i=1}^n \left( X_i - \mathbb{E} \left[ X_i \right] \right) \right\|_\infty \\
=& \max_{y \in \mathbb{S}^d} \left| \langle y| \frac{1}{n} \sum_{i=1}^n \left( X_i - \mathbb{E} \left[ X_i \right] \right) |y \rangle \right|.
\end{align*}
Next, we employ a result that is somewhat folklore in random matrix theory, see e.g.\ \cite[Lemma~5.3]{vershynin_introduction_2012}. It states that the maximum over the entire unit sphere may be replaced by a maximum over certain finite point sets, called \emph{covering nets}: A covering-net of $\mathbb{S}^d$ with fineness $\theta >0$ is a finite set of unit vectors $\left\{ z_j \right\}_{j=1}^N \subseteq \mathbb{S}^d$ that covers the entire (complex) unit sphere in the sense that every $y \in \mathbb{S}^d$ is at least $\theta$-close to a point in the net. 

\begin{lemma}
Let $\mathcal{N}_\theta = \left\{z_j \right\}_{j=1}^N$ be a covering net of $\mathbb{S}^d$ with fineness $\theta$. Then, for any matrix $A \in \mathbb{H}_d$:
\begin{equation*}
\max_{j \in \left[N\right]} \left| \langle z_j| A |z_j \rangle \right| \leq \| A \|_\infty \leq \frac{1}{1-2 \theta} \max_{j \in \left[N \right]} \left| \langle z_j | A |z_j \rangle \right|.
\end{equation*}
\end{lemma}

This result highlights that already a rather coarse net suffices to get reasonable approximations to the operator norm. Here, we choose $\theta = \frac{1}{4}$ which, while certainly not optimal, simplifies exposition. In particular,
\begin{align}
\left\| \hat{L}_n - \rho \right\|_\infty
\leq 2 \max_{j \in \left[N\right]} \left| \frac{1}{n} \sum_{j=1}^n \langle z_j| X_i -\mathbb{E} \left[ X_i \right] |z_j \rangle  \right|, \label{eq:uniform_aux1}
\end{align}
where the maximization is over a covering net of fineness $\theta = \frac{1}{4}$.

\subsection{Step II: concentration}

Note that the right hand side of Eq.~\eqref{eq:uniform_aux1} corresponds to a maximum over $N$ different random variables -- each of them labeled by a unit vector $z_j$ in the net. Let $z \in \mathbb{S}^d$ be such a vector. Then, the associated random variable itself corresponds to an empirical average of $n$ i.i.d. variables:
\begin{equation*}
s_z = \langle z| X - \mathbb{E} \left[ X \right] |z \rangle.
\end{equation*}
Clearly, $s_z$ obeys $\mathbb{E} \left[ s_z \right] =0$ and, more importantly, 
has sub-exponential moment growth. While this follows directly from the fact that $s_z$ is bounded, the following result highlights that this  tail-behavior is actually independent of the ambient dimension.

\begin{lemma} \label{lem:moments}
Fix $z \in \mathbb{S}^d$. Then for any integer $p \geq 2$, the random variable $s_z$ obeys
\begin{equation*}
\mathbb{E} \left[ |s_z |^p \right] \leq  27 \times 6^{p-2} p! 
\end{equation*}
\end{lemma}

We divert the proof of this statement to the end of this section and content ourselves with emphasizing that the closed form expression of the frame operator \eqref{eq:frame_operator} is essential for bounding all moments simultaneously.
More relevant to the task at hand is that such a moment behavior ensures that the tails of the distribution of $s_z$ follow an exponential decay: $\mathrm{Pr} \left[ |s_z| \geq t \right] \leq \mathrm{e}^{-ct}$, where $c$ is a constant independent of the dimension $d$. Strong classical concentration inequalities apply for sums of i.i.d. random variables that exhibit such sub-exponential behavior. We choose to apply a rather general version of the classical Bernstein inequality, see e.g.\ \cite[Theorem~7.30]{foucart_mathematical_2013}.

\begin{theorem} \label{thm:scalar_bernstein}
Let $s_1,\ldots,s_n \in \mathbb{R}$ i.i.d.\ copies of a mean-zero random variable $s$ that obeys $\mathbb{E} \left[ |s|^p \right] \leq p! R^{p-2} \sigma^2/2$ for all integers $p \geq 2$, where $R,\sigma^2>0$ are constants. Then, for all $t >0$,
\begin{equation*}
\mathrm{Pr} \left[ \left| \sum_{i=1}^n s_i \right| \geq t \right] \leq 
2 \exp \left( - \frac{t^2/2}{n\sigma^2 +Rt} \right).
\end{equation*} 
\end{theorem}

Lemma~\ref{lem:moments} ensures that the random variable $s_z$ meets this requirement with $\sigma^2 = 54$ and $R = 6$. Hence, the following Corollary is an immediate consequence of Theorem~\ref{thm:scalar_bernstein}.

\begin{corollary} \label{cor:scalar_concentration}
Fix $z \in \mathbb{S}^d$. Then, for any $t \in [0,1]$
\begin{equation*}
\mathrm{Pr} \left[ \left| \frac{1}{n} \sum_{i=1}^n \langle z| X_i - \mathbb{E} \left[ X_i \right] |z \rangle \right| \geq t \right]
\leq 2 \mathrm{e}^{-\frac{nt^2}{120}}.
\end{equation*}
\end{corollary}

\subsection{Step III: union bound}

Recall that Eq.~\eqref{eq:uniform_aux1} upper-bounds $\| \hat{L}_n-\rho \|_\infty$ by a maximum over finitely many random variables, each of which is controlled by the strong exponential tail inequality from Corollary~\ref{cor:scalar_concentration}. 
To exploit this, we fix $\tau \in \left[0,1 \right]$ and apply a union bound (also known as Boole's inequality) over all these different random variables to obtain
\begin{align*}
& \mathrm{Pr} \left[ \left\| \hat{L}_n - \rho \right\|_\infty \geq \tau \right] \\
\leq & \mathrm{Pr} \left[ \max_{j \in \left[N \right]} \left| \frac{1}{n} \sum_{j=1}^n \langle z_j |X_i - \mathbb{E} \left[ X_i \right] |z_j \rangle \right| \geq \frac{\tau}{2} \right] \\
\leq &N \max_{j \in \left[N \right]} \mathrm{Pr} \left[  \left|\frac{1}{n} \sum_{j=1}^n \langle z_j |X_i - \mathbb{E} \left[ X_i \right] |z_j \rangle \right| \geq \frac{\tau}{2} \right] 
\leq  2 N \mathrm{e}^{-\frac{n \tau^2}{480}},
\end{align*}
where the last line is due to Corollary~\ref{cor:scalar_concentration}
Here, $N = |\mathcal{N}_{\frac{1}{4}}|$ denotes the cardinality of a covering net for the complex unit sphere $\mathbb{S}^d$ with fineness $\theta = \frac{1}{4}$. 
The complex unit sphere admits an isometric embedding into the real-valued unit sphere in $2d$-dimensions: Map real- and imaginary parts of each complex vector component onto two distinct real parameters. This map preserves Euclidean lengths and, by extension, also the geometry of the unit sphere. Volumetric upper bounds on the cardinality of covering nets for the $2d$-dimensional real-valued unit sphere are widely known, see e.g.\ \cite[Proposition~C.3]{foucart_mathematical_2013} and \cite[Lemma~5.2]{vershynin_introduction_2012}: 
 $| \mathcal{N}_\theta | \leq \left( 1+ \frac{2}{\theta} \right)^{2d}$.
Since a fineness of $\theta = \frac{1}{4}$ suffices for our purpose, we can conclude $N \leq 3^{2d}$ and consequently,
\begin{equation*}
\mathrm{Pr} \left[ \left\| \hat{L}_n - \rho \right\|_\infty \geq \tau \right]
\leq 2 \times 3^{2d} \mathrm{e}^{- \frac{n \tau^2}{480}}
= 2 \mathrm{e}^{ 2 \log (3)d - \frac{n \tau^2}{480} }
\end{equation*}
This concludes the proof of Theorem~\ref{thm:uniform}.

\subsection{Proof of Lemma~\ref{lem:moments}}

Recall that, by assumption, the random matrix $X$ assumes the value $X= (d+1)|v \rangle \! \langle v| - \mathbb{I}$ with probability $\langle v| \rho |v \rangle \mathrm{d}v$, where $v$ may range over the entire complex unit sphere $\mathbb{S}^d$. Moreover, $\mathbb{E} \left[ X \right] = \rho$. For fixed $z \in \mathbb{S}^d$, we may therefore write
\begin{align*}
s_z =& \langle z| X - \mathbb{E} \left[ X \right] |z \rangle 
= 
 (d+1) \langle v| B |v \rangle,
\end{align*}
where $B = |z \rangle \! \langle z| - \frac{1+\langle z| \rho |z \rangle}{d+1} \mathbb{I} \in \mathbb{H}_d$ has bounded trace norm
\begin{align}
\| B \|_1 \leq 1+ \left(1+\langle z| \rho |z \rangle \right) \leq 3.
\label{eq:uniform_trace-norm_bound}
\end{align}
Next, recall a basic identity from matrix analysis that states 
\begin{equation*}
| \langle v| B| v \rangle| = \left| \mathrm{tr} \left( |v \rangle \! \langle v| B \right) \right|
\leq \mathrm{tr} \left( |v \rangle \! \langle v| \; |B| \right),
\end{equation*}
where $|B|=\sqrt{B^2}$ denotes the absolute value of the matrix $B$. 
Also, the Schatten-$p$ norms of matrices and their absolute values coincides, in particular $\| B \|_1 = \mathrm{tr}(|B|) = \| \; | B |\; \|_1$.
We can use this trick to absorb the absolute value in the moment computation. More precisely, fix an integer $p \geq 2$ and note that $\mathbb{E} \left[ |s_z |^p \right]$ obeys
\begin{align*}
 \mathbb{E} \left[ \left|(d+1) \langle v| B |v \rangle \right|^p \right] 
\leq  (d+1)^p\mathbb{E} \left[ \mathrm{tr} \left( |v \rangle \! \langle v|\; |B| \right)^p \right].
\end{align*}
We can now include the distribution of the random matrices $X$, and -- by extension -- $|v \rangle \! \langle v|$ -- to compute
\begin{align*}
\mathbb{E} \left[ |s_z|^p \right] \leq & (d+1)^p\mathbb{E} \left[ \mathrm{tr} \left( |v \rangle \! \langle v| \; |B| \right)^p \right] \\
=& d (d+1)^p \int_{\mathbb{S}^d} \langle v| \rho |v \rangle \mathrm{tr} \left( |v \rangle \! \langle v| \; |B| \right)^p \mathrm{d}v \\
=& d(d+1)^p \mathrm{tr} \left( \int_{\mathbb{S}^d} \left(|v \rangle \! \langle v|) \right)^{\otimes (p+1)} \; \rho \otimes |B|^{\otimes p} \right) \\
=& d(d+1)^p \binom{d+p}{p+1}^{-1} \mathrm{tr} \left( P_{\mathrm{Sym}^{(p+1)}} \rho \otimes |B|^{\otimes p} \right),
\end{align*}
where the last equation is due to Eq.~\eqref{eq:frame_operator}.
Next, we note that Hoelder's inequality implies
\begin{align*}
\mathrm{tr} \left( P_{\mathrm{Sym}^{(p+1)}} \rho \otimes |B|^{\otimes p} \right)
\leq \| P_{\mathrm{Sym}^{p+1}} \|_\infty \| \rho \|_1 \| \; |B| \; \|_1^p \leq 3^p,
\end{align*}
because $P_{\mathrm{Sym}^{(p+1)}}$ is an orthogonal projector, $\rho$ is a quantum state
and $B$ is bounded in trace norm \eqref{eq:uniform_trace-norm_bound}. For the remaining pre-factor we use the crude bound
\begin{equation*}
d(d+1)^p \binom{d+p}{p+1}^{-1} \leq (p+1)! \leq 3 \times 2^{p-2} p!
\end{equation*}
to establish the statement.

\section{Additional numerical experiments} \label{sec:numerics}

Maximal sets of mutually unbiased bases (MUBs) form a structured POVM (2-design) that lends itself to numerical investigation. Efficient algebraic constructions of MUBs exist in prime power dimensions $d=p^k$ \cite{wootters_optimal_1989, bandyopadhyay_new_2002,klappenecker_constructions_2004}. 
To further underline the implicit advantage of low-rank we fix a prime dimension $d$ and choose a pure state uniformly from the Haar measure on the complex unit sphere in $d$ dimensions.
We compute the outcome probabilities  for each of the $d+1$ different MUB measurements. We then sample outcomes from each distribution a total of $\frac{n}{d+1}$ times
and compute the estimator $\hat{\rho}_{n}$ associated with the total frequency statistics. 
Figure~\ref{fig:mubs} shows the relation between reconstruction error (in trace distance) and the number of samples per basis on a $\log-\log$-scale for different prime dimensions between $d=100$ and $d=200$.
This figure suggests that the rate of convergence only depends linearly on the ambient dimension -- the additional $\log (d)$-factor in the main result for structured POVMs is barely visible.

\begin{figure} 
\includegraphics[width=0.48\textwidth]{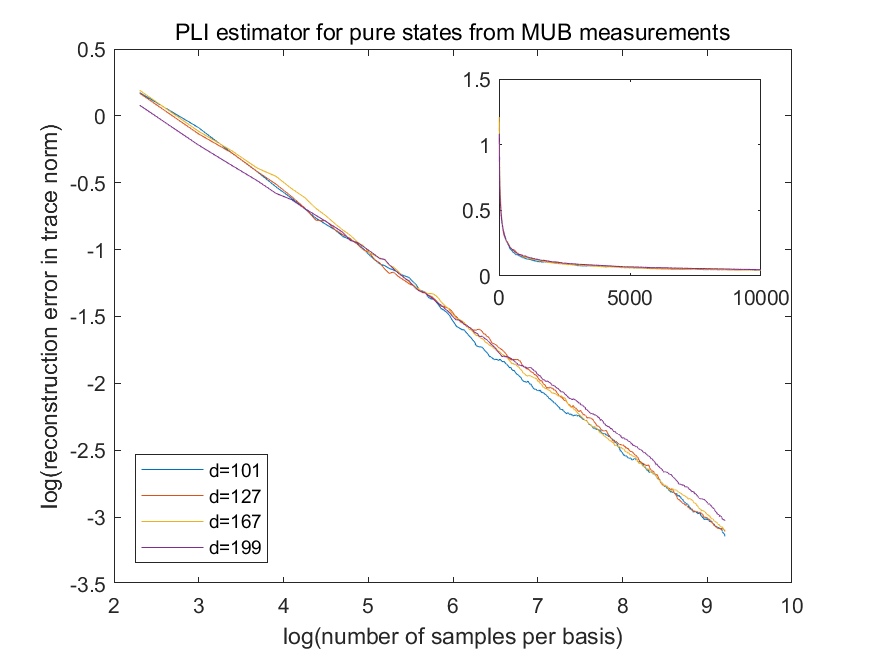}
\caption{$\log$ trace distance error vs. $\log $ sample size for different prime dimensions $d$. \emph{Inset:} ordinary plot of the same data.} \label{fig:mubs}
\end{figure}

\end{document}